\documentclass[lettersize,journal]{IEEEtran}
\usepackage{amsmath,amsfonts}
\usepackage{algorithmic}
\usepackage{algorithm}
\usepackage{array}
\usepackage{textcomp}
\usepackage{stfloats}
\usepackage{url}
\usepackage{verbatim}
\usepackage{graphicx}
\usepackage{cite}
\usepackage{subcaption}
\usepackage{framed}
\newtheorem{theorem}{Theorem}
\newtheorem{definition}{Definition} 

\newtheorem{proof}{Proof}
\usepackage{booktabs}
\usepackage{multirow}
\usepackage{tabularx}
\usepackage{makecell}
\usepackage{threeparttable}
\usepackage{stmaryrd}

\begin{document}

\title{Client-Verifiable and Efficient Federated Unlearning \\in Low-Altitude Wireless Networks}

\author{Yuhua Xu, Mingtao Jiang, Chenfei Hu, Yinglong Wang,~\IEEEmembership{Senior Member,~IEEE,} \\Chuan Zhang,~\IEEEmembership{Member,~IEEE,} Meng Li,~\IEEEmembership{Senior Member,~IEEE,} Ming Lu, Liehuang Zhu,~\IEEEmembership{Senior Member,~IEEE}
\thanks{This work was supported by the National Natural Science Foundation of China (Grant No. 62472032), and the Young Elite Scientists Sponsorship Program by CAST (Grant No. 2023QNRC001), and the Taishan Scholars Program (Grant No. tspd20240814).}
\thanks{Yuhua Xu, Mingtao Jiang, Chenfei Hu, and Liehuang Zhu are with the School of Cyberspace Science and Technology, Beijing Institute of Technology, Beijing 100081, China (e-mail: yuhuax21@bit.edu.cn; jiangmingtao@bit.edu.cn; chenfeih@bit.edu.cn; liehuangz@bit.edu.cn).}
\thanks{Yinglong Wang is with the Key Laboratory of Computing Power Network and Information Security, Ministry of Education, Shandong Computer Science Center, Qilu University of Technology (Shandong Academy of Sciences), Jinan 250014, China (e-mail: wangyinglong@qlu.edu.cn).}
\thanks{Chuan Zhang is with the School of Cyberspace Science and Technology, Beijing Institute of Technology, Beijing 100081, China, also with the Key Laboratory of Computing Power Network and Information Security, Ministry of Education, Shandong Computer Science Center, Qilu University of Technology (Shandong Academy of Sciences), Jinan 250014, China (e-mail:  chuanz@bit.edu.cn).}
\thanks{Meng Li is with the School of Computer Science and Information Engineering, Hefei University of Technology, Hefei, Anhui 230601, China (e-mail: mengli@hfut.edu.cn).}
\thanks{Ming Lu is with the Zhongguancun Rongzhi Enterprise Management Innovation Promotion Center, Beijing 100048, China (e-mail: luming@buaa.edu.cn).}
\thanks{Corresponding author: Chuan Zhang and Ming Lu.}
}

\markboth{Journal of \LaTeX\ Class Files,~Vol.~14, No.~8, August~2021}%
{Shell \MakeLowercase{\textit{et al.}}: A Sample Article Using IEEEtran.cls for IEEE Journals}


\maketitle

\begin{abstract}
In low-altitude wireless networks (LAWN), federated learning (FL) enables collaborative intelligence among unmanned aerial vehicles (UAVs) and integrated sensing and communication (ISAC) devices while keeping raw sensing data local. Due to the ``right to be forgotten'' requirements and the high mobility of ISAC devices that frequently enter or leave the coverage region of UAV-assisted servers, the influence of departing devices must be removed from trained models. This necessity motivates the adoption of federated unlearning (FUL) to eliminate historical device contributions from the global model in LAWN. However, existing FUL approaches implicitly assume that the UAV-assisted server executes unlearning operations honestly. Without client-verifiable guarantees, an untrusted server may retain residual device information, leading to potential privacy leakage and undermining trust. To address this issue, we propose VerFU, a privacy-preserving and client-verifiable federated unlearning framework designed for LAWN. It empowers ISAC devices to validate the server-side unlearning operations without relying on original data samples. By integrating linear homomorphic hash (LHH) with commitment schemes, VerFU constructs tamper-proof records of historical updates. ISAC devices ensure the integrity of unlearning results by verifying decommitment parameters and utilizing the linear composability of LHH to check whether the global model accurately removes their historical contributions. Furthermore, VerFU is capable of efficiently processing parallel unlearning requests and verification from multiple ISAC devices. Experimental results demonstrate that our framework efficiently preserves model utility post-unlearning while maintaining low communication and verification overhead.
\end{abstract}

\begin{IEEEkeywords}
Low-altitude wireless networks, federated learning, verifiable unlearning, privacy preservation, linear homomorphic hash, efficient validation.
\end{IEEEkeywords}

\section{Introduction}
\IEEEPARstart{L}{ow-altitude} wireless networks (LAWN) constitute a key layer of digital infrastructure to support unmanned aerial vehicles (UAVs), drone swarms, and future urban air mobility systems \cite{yuan2025ground, luo2025toward}. By integrating communication, sensing, control, and computing across aerial and terrestrial nodes, LAWN enables a wide range of mission-driven and safety-critical applications, including autonomous aerial coordination, environmental monitoring, emergency response, and infrastructure inspection \cite{wang2025toward}. Due to limited transmission capacity, stringent latency constraints, and the sensitive nature of sensing data, centralized model training is often impractical in LAWN. As a distributed machine learning paradigm, federated learning (FL) is a promising solution in LAWN that enables multiple clients to train models locally on their private datasets while collaboratively constructing a shared global model by exchanging model updates such as parameters or gradients \cite{mcmahan2017communication, yu2025lightweight}. In typical LAWN-based FL deployments, model aggregation is performed by UAV-assisted servers, while data are generated by integrated sensing and communication (ISAC) devices such as vehicles, cameras, and mobile phones, which act as clients in the FL process. This paradigm is well aligned with the decentralized and resource-constrained nature of LAWN. By keeping raw sensing data local, FL supports scalable learning across mobile nodes, thereby serving as a fundamental building block for edge intelligence in low-altitude networks.

However, regulations such as the General Data Protection Regulation (GDPR) \cite{voigt2017eu} and the California Consumer Privacy Act (CCPA) \cite{harding2019understanding} mandate the Right to Be Forgotten (RTBF), granting participants the authority to request the deletion of their historical data contributions. Consequently, deploying FL in LAWN needs to operate under strict regulatory constraints. Low-altitude environments exhibit frequent node turnover caused by mobility, task reassignment, maintenance, or security incidents. For example, vehicles dynamically enter or leave the coverage region of a UAV-assisted server. Within LAWN learning systems, the RTBF requirement translates into the ability to eliminate the influence of withdrawn or departing ISAC devices from learned models, ensuring that distributed intelligence remains controllable, compliant, and trustworthy.

To address this requirement, federated unlearning (FUL) has been introduced \cite{liu2021federaser, wu2022federated, ding2024strategic}. It ensures that when some clients withdraw, their historical data contributions are erased from the global model. The most straightforward approach to machine unlearning involves complete model retraining on residual data \cite{bourtoule2021machine, shaik2024exploring}. However, this incurs substantial computational overhead \cite{ginart2019making}. In LAWN-based FUL, where the server lacks direct access to client data and client participation is inherently dynamic, such retraining strategies cannot be applied directly. As a result, recent research on federated unlearning primarily focuses on eliminating the impact of historical gradient contributions from clients on the server side \cite{liu2021revfrf, guo2023fast, fu2024client}. Liu et al. \cite{liu2021federaser} proposed a method that utilizes stored historical parameter updates on the server to iteratively calibrate and reverse model contributions associated with the target client. Building on this idea, Guo et al. \cite{guo2023fast} introduced the FAST framework, which specifically targets the unlearning of malicious clients. By removing the historical updates of such clients and retraining the model on a reference dataset to correct potential bias, the framework effectively removes the contributions of malicious participants with high efficiency.

Despite their effectiveness, existing methods \cite{cao2023fedrecover, yuan2023federated, zhang2023fedrecovery} are grounded in a critical assumption: the server will faithfully execute the unlearning process. In practical LAWN deployments, the UAV-assisted servers may be operated by third-party service providers. They may deliberately retain the gradient contributions of high-value clients to maintain model utility, leading to incomplete unlearning and potential manipulation of the outcome. Residual contributions can still be exploited through inversion attacks, thereby exposing sensitive information \cite{romandini2024federated}. Furthermore, as the unlearning process is entirely controlled by the server, clients are deprived of any means to verify whether their contributions have truly been erased. The lack of independent verification mechanisms exposes a critical vulnerability in existing FUL systems and undermines trust in distributed intelligence for LAWN.

\IEEEpubidadjcol

For verifiable machine unlearning, existing methods \cite{sommer2022athena, guo2023verifying} typically require data owners to actively inject backdoor triggers into their data, such as modifying original labels to erroneous ones. If the predictions on the unlearned model retain high accuracy with original labels while failing to respond to triggered backdoor samples, it indicates that the model provider has faithfully removed the backdoor-injected data. In FL, VERIFI \cite{gao2024verifi} pioneered a verification framework based on marked data samples. In this method, clients upload specific samples and fine-tune the global model. The change in model loss is then monitored to infer whether the corresponding data contributions have been removed. Although this method eliminates the need for backdoor injection, it suffers from two fundamental limitations when applied to LAWN deployments. First, the verification process relies on access to original training samples, which is impractical in LAWN scenarios where sensing data are generated continuously and stored temporarily. Second, the method processes unlearning requests in a sequential manner, making it inefficient in low-altitude networks where many ISAC devices may concurrently enter or leave the coverage region of an UAV-assisted server.

To address the above challenges, we propose VerFU, a privacy-preserving and client-verifiable federated unlearning framework designed for low-altitude wireless networks. VerFU enables ISAC devices to verify whether the untrusted UAV-assisted server has honestly erased their historical contributions. In addition, VerFU supports lightweight verification without relying on original sensing data samples, handles concurrent unlearning requests from multiple ISAC devices, and preserves the performance of the global model after unlearning. Specifically, we design tamper-proof records of historical contributions using a synergistic combination of linear homomorphic hash and commitment scheme. ISAC devices encrypt their updates using homomorphic encryption, allowing the server to perform aggregation and unlearning directly in the encrypted domain, thus ensuring data privacy. Furthermore, leveraging the consistency of homomorphic hash linear combinations, ISAC devices can jointly verify the correctness of the unlearning result. A client state tagging mechanism is introduced to efficiently manage parallel unlearning requests and verifications from multiple ISAC devices.

The main contributions are summarized as follows:
\begin{itemize}
    \item We propose a client-verifiable federated unlearning framework for LAWN systems, which is independent of data samples. It enables ISAC devices in LAWN to verify the honesty of unlearning operations performed by untrusted UAV-assisted servers. By constructing tamper-proof records of historical contributions, the framework prevents UAV-assisted servers from forging or partially executing unlearning while preserving post-unlearning model utility.
    \item We design an efficient concurrent unlearning and verification mechanism based on a client state tagging strategy, allowing multiple unlearning requests to be processed in parallel. This mechanism accommodates frequent ISAC device join-leave behavior in LAWN systems while producing low computational and communication overhead.
    \item We develop a privacy-preserving aggregation and unlearning protocol that operates directly in the encrypted domain, preventing the leakage of raw model updates and sensitive information. We further provide rigorous security proofs and analysis to demonstrate the privacy and integrity guarantees of the proposed framework.
    \item We conduct extensive evaluations of the VerFU framework under dynamic unlearning scenarios. Experimental results demonstrate that VerFU can effectively recover model utility after unlearning while incurring minimal communication and verification overhead, validating its practicality for LAWN deployments.
\end{itemize}

The remainder of this paper is organized as follows. We present some preliminaries in Section \ref{sec2}. Subsequently, we describe the architecture, threat model, and design goals in Section \ref{sec3}. The design details of VerFU are illustrated in Section \ref{sec4}. Section \ref{sec5} provides formal security proofs. The performance evaluation of VerFU is demonstrated in Section \ref{sec6}. In Section \ref{sec7}, we review the related work. We conclude the paper in Section \ref{sec8}.

\section{Preliminaries}\label{sec2}

\subsection{Federated Unlearning}
Federated unlearning (FUL), an extension of machine unlearning, addresses the challenges of data erasure in FL settings \cite{wang2024fedu, pan2025federated}. The primary objective of federated unlearning is to ensure that the model after unlearning has the same distribution as the model retrained from scratch using the remaining clients. This requires eliminating the influence of the clients requesting unlearning on the global model \cite{liu2024survey}.

In an FL system, the participants consist of a server $S$ and a set of $N$ clients $U=\{u_{1}, \dots, u_{N}\}$. Each client $u_{i} \in U$ holds a private local dataset $D_{i}$ and collaboratively trains a global model that minimizes the average loss across all clients. Specifically, the global model at round $t$ is denoted as $w_{t}$. The client $u_{i}$ downloads the global model $w_{t}$ and trains a local model $w_{t+1}^{(i)}$ based on its respective local dataset $D_{i}$. Subsequently, the client $u_{i}$ uploads its model update, represented as gradient $v_{t+1}^{(i)}= w_{t+1}^{(i)}-w_{t}$. The server aggregates client updates and performs the global model update using the classical FedAvg algorithm \cite{mcmahan2017communication}, where the global model is updated as $w_{t+1}= w_{t}+\frac{1}{N} \sum_{u_{i}\in U}v_{t+1}^{(i)}$.
As for traditional FUL, the process is equivalent to retraining the model from scratch using only the clients that did not request unlearning. The client set $U$ is divided into two distinct subsets: the unlearning client set $U_{unl}$ and the normal client set $U_{nor}$, where $U=U_{unl} \cup U_{nor}$. Accordingly, the federated unlearning problem can be formulated as $w_{t+1}= w_{t}+\frac{1}{|U_{nor}|} \sum_{u_{i}\in U_{nor}}v_{t+1}^{(i)}$.


\subsection{Cryptographic Primitives}
\subsubsection{Homomorphic Encryption}
Homomorphic Encryption (HE) \cite{rivest1978data} performs computations directly on encrypted data, yielding decrypted results equivalent to those from plaintext computations. The Paillier cryptosystem \cite{paillier1999public} is typically defined by three core algorithms:

\begin{itemize}
\item{$\mathbf{HE.Gen}(\kappa) \to (pk,sk)$}: Given the security parameter $\kappa$, two large prime numbers $p$ and $q$ are selected such that $gcd(pq,(p-1)(q-1))=1$, where $gcd(\cdot)$ denotes the greatest common divisor function. Choose a random integer $g \in \mathbb{Z}_{n^2}^\ast$ and compute $n=pq$, forming the public key $pk=(n,g)$. The private key $sk=(\lambda,\mu)$ is calculated by $\lambda=lcm(p-1,q-1)$ and $\mu=(L(g^\lambda\mod n^2))^{-1} \mod n$, where $lcm(\cdot)$ represents the least common multiple function and $L(x)=(x-1)/n$. 
\item{$\mathbf{HE.Enc}(pk,c) \to \llbracket c \rrbracket$}: For encrypting the plaintext $c$ with the public key $pk$, select a random number $r \in \mathbb{Z}_m^\ast$ and calculate the ciphertext $\llbracket c \rrbracket = g^{c} r^{n}\mod n^2$.
\item{$\mathbf{HE.Dec}(sk,\llbracket c \rrbracket) \to c$}: With the private key $sk$, the ciphertext $\llbracket c \rrbracket$ can be decrypted through $c=L({\llbracket c \rrbracket}^\lambda \mod n^2)\mu \mod n$.
\end{itemize}
Therefore, given the plaintexts $c_1$ and $c_2$, the additive homomorphism of the Paillier cryptosystem allows computing $\llbracket c_1 +c_2 \rrbracket =\llbracket c_1 \rrbracket \cdot \llbracket c_2 \rrbracket$. This can be simplified as $\llbracket c_1 \rrbracket \oplus \llbracket c_2 \rrbracket$, where $\oplus$ represents the addition operation on the ciphertext. Besides, it supports constant multiplication by a plaintext number $b$, which is denoted as $\llbracket b \cdot c_1 \rrbracket =\llbracket c_1 \rrbracket ^b$. Subtraction can be realized through modular inverse, i.e., $\llbracket c_1 -c_2 \rrbracket =\llbracket c_1 \rrbracket \cdot \llbracket c_2 \rrbracket^{-1}$. Similarly, we simplify this operation as $\llbracket c_1 \rrbracket \ominus \llbracket c_2 \rrbracket$, where $\ominus$ represents subtraction operation on the ciphertext.


\subsubsection{Linear Homomorphic Hash}

Linear Homomorphic Hash (LHH) \cite{bellare1994incremental} is utilized for data integrity verification and consistency checks. It leverages the homomorphic properties of hash functions to detect tampering or verify correct linear operations. LHH is generally divided into three algorithms:

\begin{itemize}
\item{$\mathbf{LHH.Gen}(\kappa,d) \to pp_{LHH}$}: Taking the security parameter $\kappa$ and dimension $d$ as input, the algorithm generates the public parameters $pp_{LHH}=(G,q,g)$, where $G$ is a cyclic group of prime order $q$ with generator $g$, and $g_{i \in \lbrack d\rbrack}$ is derived based on the dimension $d$.
\item{$\mathbf{LHH.Hash}(pp_{LHH},m) \to h(m)$}: Given the public parameters $pp_{LHH}$, the linear homomorphic hash $h(m)$ of $m$ can be produced by $h(m)=\prod_{i\in \lbrack d \rbrack} g_{i}^{m \lbrack i \rbrack}$.
\item{$\mathbf{LHH.Eval}(h_{0}, \dots, h_{N-1}, f_{0}, \dots, f_{N-1}) \to H$}: With $N$ linear homomorphic hashes $h_{i \in \lbrack N \rbrack}$ and coefficients $f_{i \in \lbrack N \rbrack}$ as input, compute their combined result $H=\prod_{i \in \lbrack N \rbrack} h_{i}^{f_{i}}$ as the output.
\end{itemize}

\subsubsection{Commitment}
Commitment schemes are widely employed in information exchange, where a committer binds a value without revealing it to the verifier. The committer is unable to alter the committed value, while the verifier remains unaware before decommitment. This ensures the privacy and integrity of the data. A non-interactive equivocal commitment scheme is typically categorized into three algorithms:
\begin{itemize}
\item{$\mathbf{COM.Gen}(\kappa) \to (td, pp_{COM})$}: Given the security parameter $\kappa$ as input, the output consists of a trapdoor $td$ and the public parameters $pp_{COM}$ that define the message space and commitment space.
\item{$\mathbf{COM.Com}(pp_{COM},x,r) \to m$}: During the commitment phase, the committer generates a commitment $m$ based on the chosen random value $r$ and the message $x$, and sends it to the verifier.
\item{$\mathbf{COM.Dec}(m,x',r') \to \{0,1\}$}: During the decommitment phase, given the commitment $v$, the claimed message $x'$, and the claimed random value $r'$ as input, the verifier checks whether $m=\mathbf{COM.Com}(pp_{COM},x',r')$ holds. If the condition is satisfied, the algorithm outputs 1, indicating that the message has not been altered.
\item {$\mathbf{COM.Equ}(m,(x,r),x',td) \to r'$}: Let $m \gets \mathbf{COM.Com}(pp_{COM},x,r)$ be a commitment string derived from message $x$ and randomness $r$. Given $m$, an arbitrary target message $x'$, and a trapdoor $td$, the algorithm produces a valid randomness $r'$ such that $m$ can be decommitted to $x'$.
\end{itemize}

\section{The Architecture of Design}\label{sec3}

\subsection{Architecture Overview}
We consider a LAWN-based FL that supports collaborative intelligence. As illustrated in Fig. \ref{fig:arc}, the architecture consists of a set of UAV-assisted servers $S=\{S^{1},S^{2},\dots,S^{L}\}$, each responsible for a low-altitude coverage region, and a collection of ISAC devices operating within these regions. For each UAV-assisted server $S^l \in S$, a subset of ISAC devices $U^l=\{u^l_{1},u^l_{2},\dots,u^l_{N_l}\}$ collaboratively train a global model by performing local learning on sensed data. Due to the autonomy and mobility of ISAC devices, device participation in each region is dynamic. ISAC devices may leave a coverage region because of physical movement, task reassignment, energy constraints, or explicit privacy requests, and thus require the removal of their historical contributions from the model maintained by the corresponding server. Accordingly, ISAC devices are categorized into two operational states. Devices that leave the coverage region or explicitly request unlearning are referred to as unlearning ISAC devices, denoted by $U^l_{unl} \subseteq U^l$. The remaining ISAC devices, donated by $U^l_{nor}=U^l \setminus U^l_{unl}$ and continue participating in collaborative training. 

VerFU focuses on enabling verifiable unlearning within each server-managed region, ensuring that the contributions of unlearning ISAC devices are correctly removed from the corresponding global model. Each ISAC device maintains a private local dataset generated by its sensing tasks and contributes to the global model by uploading protected model updates. VerFU allows ISAC devices to submit unlearning requests to the associated UAV-assisted server and provides them with verifiable capabilities to confirm whether the server has correctly executed the unlearning operation. Each UAV-assisted server aggregates encrypted updates, performs unlearning upon request, and broadcasts the updated global model to devices within its coverage area. Overall, the workflow of VerFU is organized into the following three main components:

\subsubsection{Local training and Uploading}
Each ISAC device trains locally on its private dataset and protects the gradients using homomorphic encryption before uploading them to its associated UAV-assisted server through a secure communication channel. The encrypted gradients are used for model aggregation without revealing raw sensing data. In addition, each ISAC device hashes and commits its local gradients to generate tamper-proof records of historical contributions, which serve as verifiable evidence for subsequent unlearning verification.
\subsubsection{Encrypted aggregation and unlearning}
Under normal operation, each UAV-assisted server aggregates the encrypted gradients uploaded by participating ISAC devices to update the global model. When an unlearning request is issued by an unlearning ISAC device, the associated server executes the unlearning procedure to eliminate the historical contributions of that device from previous training rounds. Since all gradients are protected by homomorphic encryption, both aggregation and unlearning are performed directly in the encrypted domain, leveraging the homomorphic properties to preserve data confidentiality throughout the process.
\subsubsection{Broadcasting global model and verifying}
After completing model aggregation and unlearning, the UAV-assisted server broadcasts the updated global model to the participating ISAC devices. Normal ISAC devices directly use the received model to initiate the next round of local training. In contrast, unlearning ISAC devices first open their tamper-proof records and verify the correctness of the unlearning results by exploiting the linear composability of LHH. If verification succeeds, it confirms that the associated server has honestly executed the unlearning operation, allowing unlearning ISAC devices to safely exit the federated learning process.

\begin{figure}[tbp]
    \centering   
    \includegraphics[width=\linewidth]{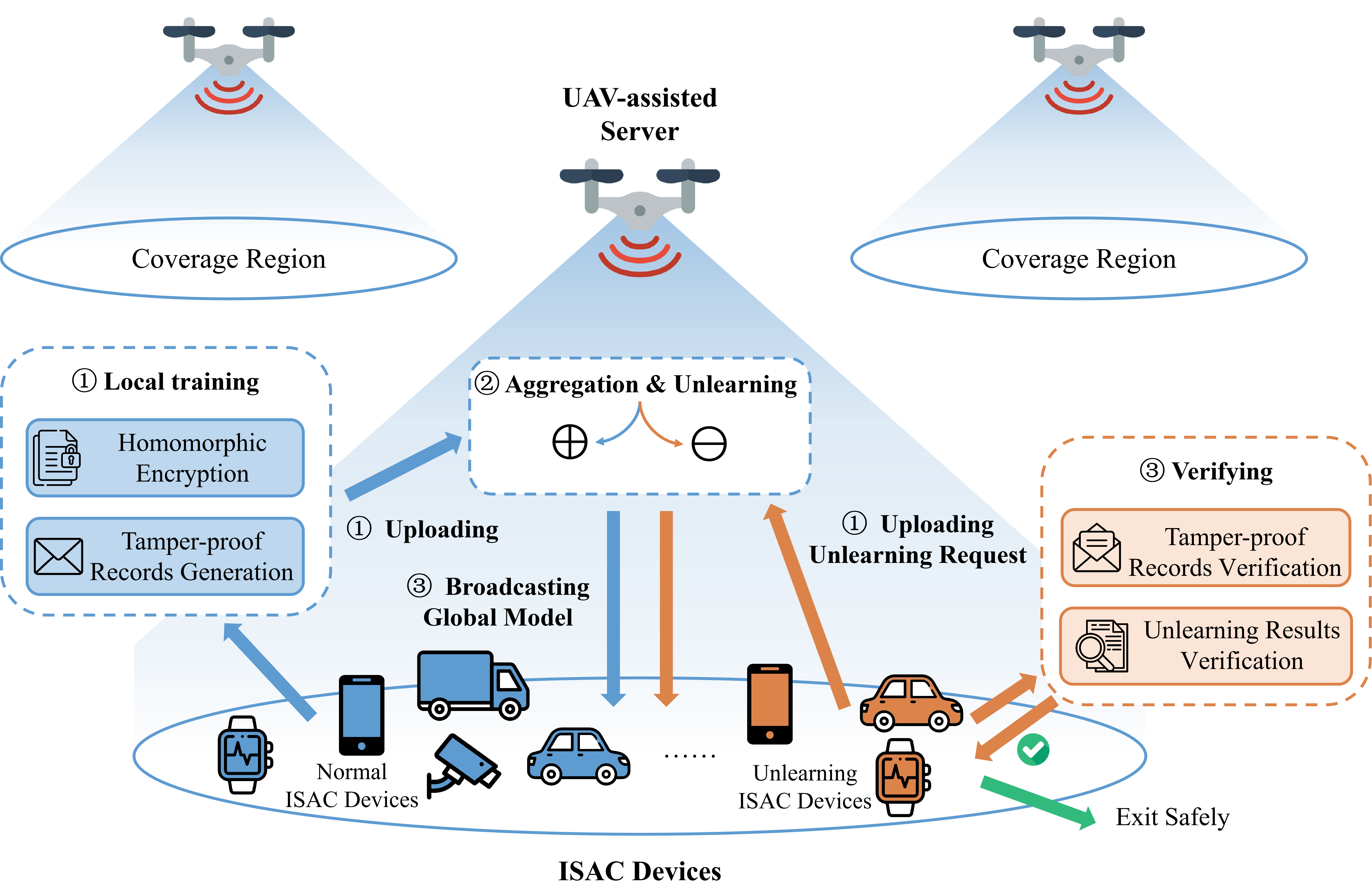}
    \caption{The architecture of VerFU.}
    \label{fig:arc}
\end{figure}

\subsection{Threat Model}
In the considered LAWN setting, we focus on potential security threats posed by untrusted UAV-assisted servers, which may be operated by third-party service providers or infrastructure operators. Each UAV-assisted server is assumed to follow the standard procedures required for FL tasks, but may behave dishonestly when handling unlearning requests issued by ISAC devices. In practice, performing unlearning incurs additional computational overhead and may degrade model performance, especially when the contributions of high-value or long-term devices are removed. As a result, a UAV-assisted server may falsely claim to have completed the unlearning operation while failing to remove, or only partially removing, the historical contributions of the requesting device. Under this threat model, residual device contributions may persist in the global model after the unlearning process, leading to privacy leakage and a loss of trust. Additionally, we assume that ISAC devices behave honestly. Therefore, the objective of VerFU is to provide verifiable FUL for LAWN deployments, enabling ISAC devices to independently verify whether a UAV-assisted server has correctly and honestly removed their historical contributions from the corresponding global model.

\subsection{Design Goals}
\subsubsection{Privacy of gradients}
The gradients generated by ISAC devices inherently reflect the characteristics of their locally sensed data. If these gradients are directly accessible to UAV-assisted servers in an unencrypted form, the servers could potentially exploit this information to infer sensitive details about the ISAC devices. Therefore, it is essential that the uploaded gradients are meaningful to the servers solely for aggregation, preserving their privacy.

\subsubsection{Effectiveness of unlearning}
When an ISAC device requests unlearning, the corresponding UAV-assisted server must completely remove contributions associated with that ISAC device from the global model. Any residual contributions could undermine the credibility and discourage future participation. Besides, the unlearning process should ensure that it does not negatively impact the training of other ISAC devices or significantly degrade the performance of the global model.

\subsubsection{Reliability of verification}
The correctness of the unlearning operation performed by UAV-assisted servers must be verified by the unlearning ISAC devices. This prevents the servers from falsifying unlearning proofs or incompletely executing the unlearning process. Moreover, the verification mechanism should be easy to implement and should not require the ISAC devices to provide additional samples to the server for comparing model performance before and after unlearning, thereby avoiding unnecessary computational overhead or additional privacy risk.

\begin{table}[tbp]
    \centering
    \renewcommand{\arraystretch}{1.2}
    \caption{Key Notations}\label{tab:4-1}
    \begin{tabular}{ll} 
    \toprule
    Notation& Description\\ 
    \midrule
    $U$& The set of all participating ISAC devices\\
    $U_{nor}$&The set of ISAC devices normally participating\\
    $U_{unl}$&The set of ISAC devices requesting to leave\\
    $u_{i}$&The $i$-th ISAC device\\
    $t$&The $t$-th iteration round\\
    $w_{t},w^{t'}$&The global model at round $t$\\ 
    $v_{t}^{(i)}$& The gradient of ISAC device $i$ at round $t$\\
    $w_{t}^{(i)}$& The local model of ISAC device $i$ at round $t$\\
    $cv_{i}$& The historical contribution of ISAC device $i$\\
    $h_{t}^{(i)}$& The LHH value of ISAC device $i$ at round $t$\\
    $r_{i}$& The random string of ISAC device $i$\\
    $m_{t}^{(i)}$& The commitment string of ISAC device $i$ at round $t$\\
    $(pk,sk)$& The key pair for homomorphic encryption\\
    $f_{i}$& The state indicator of ISAC device $i$\\
    $\Delta w^t,\Delta w^{t'}$& The average update of all ISAC devices at round $t$\\
    $a$& The aggregation result with unlearning\\
    $H$& The linear combination of $h_{t}^{(i)}$ across all ISAC devices\\
    \bottomrule
    \end{tabular}
\end{table}

\section{The proposed VerFU framework} \label{sec4}
In this section, we provide the detailed design of VerFU, a novel framework for verifiable federated unlearning in LAWN deployments. The core of VerFU is to enable UAV-assisted servers to remove the historical contributions of ISAC devices requesting unlearning during the model aggregation process, while allowing those devices to independently verify the correctness of the unlearning operation. In this way, the influence of departed or privacy-sensitive ISAC devices can be effectively eliminated from the learned model without relying on server trust. The complete framework of VerFU, as depicted in Fig. \ref{fig:full version}, is divided into three distinct phases: the preparation phase, the aggregation and unlearning phase, and the verification phase. To facilitate a clearer understanding of the design, the key notations used throughout this section are summarized in Table \ref{tab:4-1}.

\subsection{Preparations for verifiable federated unlearning}
\subsubsection{ISAC device-side training}
We consider a UAV-assisted server $S$ and its associated ISAC devices $U=\{u_{1},u_{2},\dots,u_{N}\}$. In each training round $t$, where $t \in \lbrack T \rbrack $, each ISAC device $u_{i}$ performs local training on its private dataset $D_{i}$, relying on the global model $w_{t-1}$ from the previous round. The resulting locally trained model is $w_{t}^{(i)}$, from which the gradient $v_{t}^{(i)}$ can be computed by $v_{t}^{(i)}= w_{t}^{(i)}-w_{t-1}$. 

To support verifiable federated unlearning in LAWN, each ISAC device $u_{i} \in U$ maintains a record of its historical contributions, denoted as $cv_{i}$, as follows:
\begin{equation}
    \label{eq4-0}
    cv_{i} = cv_{i}+v_{t}^{(i)}
\end{equation}
Specifically, the historical contribution for each ISAC device is initialized at the beginning round. Based on the gradients and historical contributions, normal ISAC devices and unlearning ISAC devices generate distinct linear homomorphic hashes, as illustrated below:
\begin{align}
{h_{t}^{(i)}} \leftarrow \begin{cases}
\mathbf{LHH.Hash}(pp_{LHH},v_{t}^{(i)}),&u_{i} \in U_{nor} \\ 
{\mathbf{LHH.Hash}(pp_{LHH},cv_{i}),}&u_{i} \in U_{unl} 
\end{cases}
\end{align}
Here, $pp_{LHH}$ represents the public parameters of the linear homomorphic hash, which are generated using the security parameter $\kappa$ through the $\mathbf{LHH.Gen}(\cdot)$ function. Each ISAC device further selects a random string $r_i$ to generate a commitment string $m_t^{(i)}$ based on Eq. (\ref{eq4-2}), which also serves as a tamper-proof record, and then transmits it to the associated UAV-assisted server.
\begin{equation}
    \label{eq4-2}
    m_{t}^{(i)} \leftarrow \mathbf{COM.Com}(pp_{COM},h_{t}^{(i)},r_{i})
\end{equation}
In Eq. (\ref{eq4-2}), $pp_{COM}$ is generated by $\mathbf{COM.Gen}(\cdot)$ function using a security parameter $\kappa$. Besides, it is worth noting that $(h_{t}^{(i)},r_{i})$ serves as the commitment opening string required for subsequent verification.

\subsubsection{UAV-assisted server-side broadcasting}
The UAV-assisted server $S$ first initializes the global model $w_{0}$ and broadcasts it to all ISAC devices to facilitate local training. In round $t$, upon receiving the commitment strings from the ISAC devices, the server broadcasts $\{(i,m_{t}^{(i)})\}_{u_{i} \in U}$ to all ISAC devices, ensuring the subsequent verification process can be carried out effectively.

\subsection{How to achieve aggregation with unlearning}
In a standard training process, where no ISAC devices request unlearning, the UAV-assisted server updates the global model by aggregating the gradients from all participating ISAC devices based on the global model from the previous round. Assuming the initial global model is denoted as $w_{0}$, the global model $w_{k}$ after $k$ rounds of iterative training can be expressed as follows:
\begin{equation}
    \label{eq4-4}
    w_{k}= w_{0}+\sum_{t=1}^{k} \Delta w^{t}
\end{equation}
In Eq. (\ref{eq4-4}), $\Delta w^{t}$ represents the weighted update from all participating ISAC devices during round $t$, expressed as $\Delta w^{t}=(1/|U|)\cdot \sum_{u_{i}\in U}v_{t}^{(i)}$, where $|U|$ denotes the number of ISAC devices in the set $U$.

Our framework considers scenarios in LAWN-based FL where ISAC devices dynamically leave the learning process due to mobility or privacy requirements and request the removal of their historical contributions. In this case, ISAC devices are categorized into normal ISAC devices $U_{nor}$ and unlearning ISAC devices $U_{unl}$. Accordingly, the global model parameter update $\Delta w^{t}$ for each round can be computed by separately aggregating the gradients of these two groups, as expressed below:
\begin{equation}
    \label{eq4-5}
    \Delta w^{t}=\frac{1}{|U|} \sum_{u_{i}\in U_{nor}}v_{t}^{(i)}+\frac{1}{|U|} \sum_{u_{i}\in U_{unl}}v_{t}^{(i)}
\end{equation}
Furthermore, to eliminate the historical updates of unlearning ISAC devices, the parameter updates of the global model at each round should exclusively aggregate the gradients of the normal ISAC devices, as expressed below:
\begin{align}
    \label{eq4-6}
    \Delta w^{t’}&=\frac{1}{|U_{nor}|} \sum_{u_{i}\in U_{nor}}v_{t}^{(i)}\notag
    \\&=\frac{1}{|U|-|U_{unl}|}(\sum_{u_{i}\in U}v_{t}^{(i)}- \sum_{u_{i}\in U_{unl}}v_{t}^{(i)})\notag
    \\&=\frac{|U|}{|U|-|U_{unl}|}\Delta w^{t}-\frac{1}{|U|-|U_{unl}|} \sum_{u_{i}\in U_{unl}}v_{t}^{(i)}
\end{align}
Naturally, from the derivation in Eq. (\ref{eq4-6}), $\Delta w^{t}$ can also be interpreted as subtracting the local model updates of the unlearning ISAC devices, followed by a proportional adjustment.

Moreover, observing the above equation reveals that when a significant number of ISAC devices request unlearning in the $t$-th round, but their contributions are relatively small, the second term approximately equals zero. Consequently, the first term is amplified due to the adjustment factor ${|U|}/{(|U|-|U_{unl}|)}$, which magnifies the global model update. This unintended impact on the performance of the global model requires further mitigation. Therefore, our framework assumes that unlearning ISAC devices participate in the training process but are not required to provide their gradients in the round $t$. The global model update $\Delta w^{t'}$ can be reformulated as follows:
\begin{align}
    \label{eq4-7}
    \Delta w^{t’}&=\frac{1}{|U|} \sum_{u_{i}\in U_{nor}}v_{t}^{(i)}\notag
    \\&=\frac{1}{|U|}(\sum_{u_{i}\in U}v_{t}^{(i)}- \sum_{u_{i}\in U_{unl}}v_{t}^{(i)})\notag
    \\&=\Delta w^{t}-\frac{1}{|U|} \sum_{u_{i}\in U_{unl}}v_{t}^{(i)}
\end{align}
Unlike Eq. (\ref{eq4-4}), the global model $w_{k}^{'}$ after $k$ rounds is updated based on the global model update $\Delta w^{t'}$, using an aggregation mechanism that incorporates unlearning, as follows:
\begin{align}
    \label{eq4-8}
    w_{k}^{'}&=w_{0}+\sum_{t=1}^{k} \Delta w^{t’}\notag
    \\&=w_{0}+\sum_{t=1}^{k}\Delta w^{t}-\frac{1}{|U|}\sum_{t=1}^{k} \sum_{u_{i}\in U_{unl}}v_{t}^{(i)}\notag
    \\&=w_{k-1}+\Delta w^{k}-\frac{1}{|U|} \sum_{u_{i}\in U_{unl}}cv_{i}\notag
    \\&=w_{k-1}+\frac{1}{|U|} \sum_{u_{i} \in U_{nor}} v_{k}^{(i)}-\frac{1}{|U|} \sum_{u_{i}\in U_{unl}}cv_{i}
\end{align}

In VerFU, to prevent the UAV-assisted server $S$ from directly accessing plaintext gradients, we utilize a homomorphic encryption algorithm to ensure the privacy of information uploaded by ISAC devices to the server. 

\subsubsection{ISAC device-side uploading}
A public-private key pair $(pk,sk)$ is first generated by $(pk,sk) \leftarrow \mathbf{HE.Gen}(\kappa)$ with a security parameter $\kappa$. Subsequently, each ISAC device $u_{i} \in U$ utilizes the public key $pk$ to encrypt its respective gradient or historical contributions as follows:
\begin{equation}
    \label{eq4-3}
    \llbracket v_{t}^{(i)} \rrbracket \leftarrow \mathbf{HE.Enc}(pk,v_{t}^{i}), u_{i} \in U_{nor}
\end{equation}
\begin{equation}
    \label{eq4-3}
    \llbracket cv_{i} \rrbracket \leftarrow \mathbf{HE.Enc}(pk,cv_{i}), u_{i} \in U_{unl}
\end{equation}

Furthermore, the subsequent model update $x_{i}$ to be uploaded is assigned the corresponding encrypted gradient or encrypted historical contribution as follows:
\begin{align}
{x_{i}} = \begin{cases}
\llbracket v_{t}^{(i)} \rrbracket,&u_{i} \in U_{nor} \\ 
{\llbracket cv_{i} \rrbracket,}&u_{i} \in U_{unl} 
\end{cases}
\end{align}

Given that ISAC devices $u_{i} \in U$ can either actively participate in training or request to leave by invoking the unlearning process, a status indicator $f_{i}$ is introduced to enable the UAV-assisted server $S$ to distinguish between these two states. Specifically, for normal ISAC devices $u_{i} \in U_{nor}$, the status is set as $f_{i} \leftarrow 1$, while for unlearning ISAC devices $u_{i} \in U_{unl}$, the status is set as $f_{i} \leftarrow -1$.

Building on the above operations, each ISAC device $u_{i} \in U$, upon receiving the broadcasted set $\{(i,m_{t}^{(i)})\}_{u_{i} \in U}$ from the server, packages its encrypted information together with the status indicator into $\{(i,x_{i},f_{i})\}_{u_{i} \in U}$, which is then uploaded to the UAV-assisted server. Notably, unlearning ISAC devices are not required to provide gradients during the round in which the unlearning request is made.

\subsubsection{UAV-assisted server-side unlearning}
Upon receiving the encrypted local model information and status indicators $\{(i,x_{i},f_{i})\}_{u_{i} \in U}$ from ISAC devices, the UAV-assisted server performs aggregation with unlearning directly on ciphertexts. However, since the adjustment factor may involve floating-point values, the aggregation and unlearning process on the server side is scaled by a factor of $|U|$, as expressed below:
\begin{align}
    \label{eq4-9}
    \llbracket a \rrbracket &=\llbracket v_{t}^{(i_{1})} \rrbracket \oplus \dots \oplus \llbracket v_{t}^{(i_{|U_{nor}|})} \rrbracket \ominus \llbracket cv_{j_{1}} \rrbracket \ominus \dots \ominus \llbracket cv_{j_{|U_{unl}|}} \rrbracket\notag
    \\&=\llbracket \sum_{u_{i} \in U_{nor}} v_{t}^{(i)} \rrbracket  \ominus  \llbracket \sum_{u_{j}\in U_{unl}}cv_{j} \rrbracket\notag
    \\&=\llbracket \sum_{u_{i} \in U_{nor}} v_{t}^{(i)}-\sum_{u_{j}\in U_{unl}}cv_{j} \rrbracket
\end{align}
Here, $\llbracket a \rrbracket$ represents the result of aggregation with unlearning for the current round, which remains in its encrypted form. The server then broadcasts this result $\llbracket a \rrbracket$ along with the status indicators $\{(i,f_{i})\}_{u_{i} \in U}$ to all ISAC devices.

\subsubsection{ISAC device-side decrypting}
Upon receiving the broadcasted information from the UAV-assisted server, ISAC devices are required to decrypt the aggregated results using their private key $sk$ as follows:
\begin{equation}
    \label{eq4-10}
    a \leftarrow \mathbf{HE.Dec}(\llbracket a \rrbracket)
\end{equation}
For normal ISAC devices, the decrypted aggregated results can be directly utilized to proceed with the next round of iterative training. In contrast, for unlearning ISAC devices, it is necessary to verify the unlearning results performed by the server to ensure a trustworthy departure.

\subsection{Validation for unlearning process}
To enable verifiable FUL in LAWN, unlearning ISAC devices must be able to independently verify both the integrity of committed contribution records and the correctness of the unlearning operation performed by the UAV-assisted server.
\subsubsection{ISAC Device-side Commitment Verifying}
Each ISAC device first upload the commitment opening string $\{h_{t}^{(i)},r_{i}\}$ to the UAV-assisted server after receiving the broadcasted result $a$. The unlearning ISAC devices then verify the consistency between the uploaded hash values and the decommitted hash values from all participating ISAC devices, as follows: 
\begin{equation}
    \label{eq4-11}
    \mathbf{COM.Dec}(m_{t}^{(i)},h_{t}^{(i)},r_{i}) \overset{?}{=} 1
\end{equation}
If a verification failure occurs for any ISAC device, it indicates that the tamper-proof record has been either incorrectly computed or maliciously forged, thereby compromising the integrity of the unlearning procedure.
\subsubsection{ISAC Device-side Unlearning Verifying}
Unlearning ISAC devices need to further verify the unlearning execution of the UAV-assisted server. Given that $f_{i}$ distinctly identifies both normal ISAC devices and unlearning ISAC devices, the unlearning process can be simulated through a linear combination of the homomorphic hashes of either current or historical gradient contributions from all participating ISAC devices, as follows:
\begin{equation}
    \label{eq4-12}
    H \leftarrow \mathbf{LHH.Eval}(h_{t}^{(1)},\dots,h_{t}^{(N)},f_{1},\dots,f_{N})
\end{equation}
It is crucial to emphasize that all computations are performed directly in the context of linear homomorphic hashes, thereby preserving the privacy of data for each ISAC device throughout the verification process.

Subsequently, the unlearning ISAC devices verify whether the linear homomorphic hash of the result $a$ matches the computed linear combination $H$, as formalized by the following equivalence relation:
\begin{equation}
    \label{eq4-13}
    \mathbf{LHH.Hash}(a) \overset{?}{=} H
\end{equation}
If this equivalence holds, it provides cryptographic evidence that the UAV-assisted server has faithfully executed the aggregation operation with proper unlearning. Thus, the unlearning ISAC devices can securely and verifiably exit the FUL process.

\begin{figure*}[htbp]
\begin{minipage}{0.99\textwidth}
    \fbox{
    \begin{minipage}{0.99\linewidth}
        \centering
        \vspace{0.25cm}
        \textbf{VerFU Framework}
        \vspace{0.25cm}
        \begin{itemize}
            \item \textbf{Preparation Phase:}
            \vspace{0.15cm}
            \\For each ISAC device $u_{i} \in U$ in parallel:
            \begin{itemize}
                \item Train the local model $w_t^{(i)}$ based on the previous global model $w_{t-1}$ and its own local dataset, and compute the gradient $v_t^{(i)} \leftarrow w_t^{(i)}-w_{t-1}$.
                \item Compute total historical contribution up to the current round: $cv_{i} = cv_{i}+v_{t}^{(i)}$.
                \item For each normal ISAC device $u_{i} \in U_{nor}$, compute the linear homomorphic hash of the gradient: ${h_{t}^{(i)}} \leftarrow \mathbf{LHH.Hash}(pp_{LHH},v_{t}^{(i)})$; For each unlearning ISAC device $u_{i} \in U_{unl}$, compute the linear homomorphic hash of the historical contribution: ${h_{t}^{(i)}} \leftarrow \mathbf{LHH.Hash}(pp_{LHH},cv_{i})$.
                \item Select a random string $r_{i}$ and compute a commitment: $m_{t}^{(i)} \leftarrow \mathbf{COM.Com}(pp_{COM},h_{t}^{(i)},r_{i})$.
                \item Upload $m_{t}^{(i)}$ to the server.
            \end{itemize}
            UAV-assisted server:
            \begin{itemize}
                \item Collect $m_{t}^{(i)}$ from all participating ISAC devices and broadcast $\{i,m_{t}^{(i)}\}_{u_{i} \in U}$ to the ISAC devices $u_{i} \in U$.
            \end{itemize}
            \vspace{0.15cm}
            
            \item \textbf{Aggregation and Unlearning Phase:}
            \vspace{0.15cm}
            \\For each ISAC device $u_{i} \in U$ in parallel:
            \begin{itemize}
                \item Upon receiving $\{i,m_{t}^{(i)}\}_{u_{i} \in U}$ by the server, normal ISAC devices encrypt gradients as $\llbracket v_{t}^{(i)} \rrbracket \leftarrow \mathbf{HE.Enc}(pk,v_{t}^{(i)})$, while unlearning ISAC devices encrypt historical contributions as $\llbracket cv_{i} \rrbracket \leftarrow \mathbf{HE.Enc}(pk,cv_{i})$.
                \item Set $f_i \leftarrow 1$ for normal ISAC devices and $f_i \leftarrow -1$ for unlearning ISAC devices.
                \item Upload $x_i$ and $f_i$ to the server, where $x_i=\llbracket v_{t}^{(i)} \rrbracket$ for normal ISAC devices and $\llbracket cv_{i} \rrbracket$ for unlearning ISAC devices.
            \end{itemize}
            UAV-assisted server:
            \begin{itemize}
                \item Collect the uploaded information $\{i,x_i,f_i\}$ from each ISAC device.
                \item When $f_i=1$, perform $\oplus$ operation on $x_i$; When $f_i=-1$, perform $\ominus$ operation on $x_i$.
                \item Compute the aggregation and unlearning result $\llbracket a \rrbracket =\llbracket v_{t}^{(i_{1})} \rrbracket \oplus \dots \oplus \llbracket v_{t}^{(i_{|U_{nor}|})} \rrbracket \ominus \llbracket cv_{j_{1}} \rrbracket \ominus \dots \ominus \llbracket cv_{j_{|U_{unl}|}} \rrbracket$ and broadcast $\llbracket a \rrbracket$ to all ISAC devices $u_i \in U$.
            \end{itemize}
            \vspace{0.15cm}
            
            \item \textbf{Verification Phase:}
            \vspace{0.15cm}
            \\For each ISAC device $u_{i} \in U$ in parallel:
            \begin{itemize}
                \item Upon receiving $\llbracket a \rrbracket$ from the server, decrypt the aggregation and unlearning result $a \leftarrow \mathbf{HE.Dec}(\llbracket a \rrbracket)$.
                \item Upload $\{h_{t}^{(i)},r_{i}\}$ to the server.
            \end{itemize}
            UAV-assisted server:
            \begin{itemize}
                \item Collect the opening string $\{i,h_{t}^{(i)},r_{i}\}$ and broadcast them to unlearning ISAC devices $u_i \in U_{unl}$.
            \end{itemize}
            For each ISAC device $u_{i} \in U_{unl}$ in parallel:
            \begin{itemize}
                \item Upon receiving $\{i,h_{t}^{(i)},r_{i}\}$, verify whether $\mathbf{COM.Dec}(m_{t}^{(j)},h_{t}^{(j)},r_{j}) = 1$ holds for all other ISAC devices $u_{j} \in U \backslash \{i\}$. If the condition is satisfied, continue; otherwise, output a verification failure.
                \item Compute the combination of the linearly homomorphic hashes from all ISAC devices as 
                \\$H \leftarrow \mathbf{LHH.Eval}(h_{t}^{(1)},\dots,h_{t}^{(N)},f_{1},\dots,f_{N})$.
                \item Verify whether $\mathbf{LHH.Hash}(a) = H$ holds. If the condition is satisfied, unlearning ISAC devices can trustfully exit; otherwise, output a verification failure.
            \end{itemize}
            \vspace{0.15cm}
    \end{itemize}
    \end{minipage}
}
\end{minipage}
\caption{The complete workflow of VerFU.}
\label{fig:full version}
\end{figure*}
\section{Security Analysis} \label{sec5}
In this section, we present the security analysis of VerFU, focusing on its ability to preserve gradient privacy while ensuring the integrity of unlearning results. The analysis is conducted under the architecture described in Section \ref{sec3}, where a set of ISAC devices collaboratively train a model under the coordination of UAV-assisted servers. We consider a set $U$ of $N$ ISAC devices and a UAV-assisted server $S$, with a primary focus on addressing potential security threats posed by an untrusted UAV-assisted server. The adversarial entity is modeled as a set $M$, including the UAV-assisted server and external attackers.

We provide formal definitions for the collision resistance of the linear homomorphic hash function, as well as the binding and equivocality properties of the commitment scheme.

\begin{definition}[Collision Resistance of Linear Homomorphic Hash]\label{def1}
Consider the following collision experiment, parameterized by a probabilistic polynomial-time (PPT) adversary $M$, a security parameter $\kappa$, and the vector dimension $d$.

$\mathbf{LHH\mbox{-}EXP}^{coll}_{M}(\kappa,d)$:
\begin{enumerate}
    \item $pp_{LHH} \leftarrow \mathbf{LHH.Gen}(\kappa,d)$
    \item $(x_1,x_2) \leftarrow M(pp_{LHH})$, where $x_1,x_2 \in \mathbb{F}_q^d$ and \\$x_1 \neq x_2$
    \item Output 1 if $\mathbf{LHH.Hash}(x_1) = \mathbf{LHH.Hash}(x_2)$, otherwise output 0.
\end{enumerate}
The LHH function satisfies collision resistance if for any PPT adversary $M$, the advantage $Adv^{coll}_{M,LHH}(\kappa,d)$ in finding a collision is bounded by a negligible function $\epsilon(\cdot)$:
\begin{equation}
    \label{eq5-1}
    Adv^{coll}_{M,COM}(\kappa,d):=\mathrm{Pr}[\mathbf{LHH\mbox{-}EXP}^{coll}_{M}(\kappa,d)=1] \leq \epsilon(\kappa) \notag
\end{equation}
\end{definition}

\begin{definition}[Binding of Commitment Scheme]\label{def2}
Given the security parameter $\kappa$, a commitment scheme is said to satisfy the binding property if, for any PPT adversary $M$, there exists a negligible function $\epsilon(\cdot)$ such that the following holds:
\begin{equation}
    \label{eq5-2}
    \Pr\left[b=1 \,\middle|\, 
    \begin{array}{l}
    (td,pp_{COM}) \leftarrow \mathbf{COM.Gen}(\kappa), \\
    (m,x,r,x',r') \leftarrow M(pp_{COM}), \\
    b \leftarrow \mathbf{COM.Dec}(m,x,r) \\\land \mathbf{COM.Dec}(m,x',r') \land x \neq x'
    \end{array}
    \right] \leq \epsilon(\kappa) \notag
\end{equation}
\end{definition}

\begin{definition}[Equivocality of Commitment Scheme]\label{def3}
Given the security parameter $\kappa$, a commitment scheme is said to satisfy the equivocality property if, for any message $x$ and uniformly random $r$, there exists a negligible function $\epsilon(\cdot)$ such that the following holds:
\begin{equation}
    \label{eq5-3}
    \Pr\left[b=0 \,\middle|\, 
    \begin{array}{l}
    (td,pp_{COM}) \leftarrow \mathbf{COM.Gen}(\kappa), \\
    m \leftarrow \mathbf{COM.Com}(x,r), \\
    x' \leftarrow \mathcal{X}, \\
    r' \leftarrow \mathbf{COM.Equ}(m, (x, r), x',td), \\
    b \leftarrow \mathbf{COM.Dec}(m,x',r')
    \end{array} 
\right] \leq \epsilon(\kappa) \notag
\end{equation}
\end{definition}

\subsection{Integrity of Unlearning}
The integrity of unlearning results ensures that, after all participating ISAC devices upload the linear homomorphic hashes of their gradients, any adversary attempting to disrupt honest unlearning can be detected with overwhelming probability. This property allows unlearning ISAC devices to verify the results with confidence and securely exit the FUL process.

\begin{definition}[Integrity of Unlearning]\label{def4}
In round $t$, let $\bigoplus_{i \in U_{nor}} v_{t}^{(i)}$ denote the ciphertext space aggregation of the encrypted gradients from normal ISAC devices $u_{i} \in U_{nor}$, and let $\bigoplus_{j \in U_{unl}} cv_{j}$ represent the ciphertext space aggregation  of the encrypted historical contributions from unlearning ISAC devices $u_{j} \in U_{unl}$. Encrypted values are uploaded during the \textbf{Aggregation and Unlearning Phase}. Additionally, $M(t, msg_{t},r_{t})$ describes the output of the adversary in round $t$, where $msg_{t}$ represents the transcript of messages observed so far, and $r_{t}$ denotes the joint randomness executed by the adversary. The integrity of the unlearning process is guaranteed if the probability that the adversary causes verifying ISAC devices to accept a forged unlearning result is negligible, as formalized by the following condition:

\begin{equation}
\label{eq5-4}
\Pr 
\begin{aligned}[t]
    & \left[ u_{i} \textit{ outputs } \perp \,\middle|\, 
    \begin{aligned}
        & u_{i} \in U_{unl}, \textit{ for round } t, \\
        & \llbracket a \rrbracket \leftarrow \Bigl(\bigoplus_{i \in U_{nor}} v_{t}^{(i)}\Bigr) \ominus \Bigl(\bigoplus_{j \in U_{unl}} cv_{j}\Bigr), \\
        & \llbracket a' \rrbracket \leftarrow M(t, msg_{t}, r_{t}), \\
        & a \leftarrow \mathbf{HE.Dec}(\llbracket a \rrbracket), \\
        & a' \leftarrow \mathbf{HE.Dec}(\llbracket a' \rrbracket), \\
        & a' \neq a
    \end{aligned} \right] \\
    & \geq 1 - \epsilon(\kappa)
\end{aligned}
\notag
\end{equation}
where $\epsilon(\cdot)$ is a negligible function, $\ominus$ is the minus operation in the ciphertext space, and $\kappa$ represents the security parameter.
\end{definition}

\begin{theorem}
Under the assumptions of the hardness of the discrete logarithm problem and the security of the commitment scheme, VerFU achieves unlearning integrity as formally defined in Definition \ref{def4}.
\end{theorem}

\begin{proof}
Assume the existence of a PPT adversary that successfully outputs a forged unlearning result $a' \neq a$ during the \textbf{Aggregation and Unlearning Phase}, meaning the forged result passes verification. For this to occur, the decommitment step must have been completed without errors. ISAC devices $u_{i} \in U$ upload their corresponding commitment opening strings $(h_{t}^{(i)},r_{i})$, which include the linear homomorphic hashes of gradients or historical contributions as well as random strings, to the UAV-assisted server. Upon receiving these strings, the server broadcasts them to the unlearning ISAC devices. However, the adversary may attempt to interfere by transmitting incorrect commitment opening strings. Due to the binding property and equivocality of the commitment scheme (Definition \ref{def2} and Definition \ref{def3}), the probability that unlearning ISAC devices fail to verify $\mathbf{COM.Dec}(m_{t}^{(j)},h_{t}^{(j)},r_{j}) = 1$ for all other ISAC devices $u_{j} \in U \backslash \{i\}$ is non-negligible. This indicates that the adversary cannot modify the valid linear homomorphic hashes uploaded by honest ISAC devices without being detected.

Furthermore, successful verification requires the unlearning result $a$ and the final hash $H$ to satisfy $\mathbf{LHH.Hash}(a)=H$. Under the condition that the linear homomorphic hash remains unaltered, the verification involves a linear combination expressed as:
\begin{align}
    \label{eq5-5}
    H&=\mathbf{LHH.Eval}(h_{t}^{(1)},\dots,h_{t}^{(N)},f_{1},\dots,f_{N})\notag\\
    &=\prod_{i \in U} {(h_{t}^{(i)})}^{f_{i}}\notag\\
    &=\frac{\prod_{i \in U_{nor}} h_{t}^{(i)}}{\prod_{j \in U_{unl}} h_{t}^{(j)}}\notag\\
    &=\frac{\prod_{k \in [d]} g_{k}^{\sum_{i \in U_{nor}}v_{t}^{(i)}[k]}}{\prod_{k \in [d]} g_{k}^{\sum_{j \in U_{unl}}cv_{j}[k]}}
\end{align}
where $v_{t}^{(i)}$ and $cv_{j}$ denote the gradients and historical contributions of normal and unlearning ISAC devices, respectively. Besides, the vector dimension is $d$. Simultaneously, the forged unlearning result $a'$ can be represented as:
\begin{align}
    \label{eq5-6}
    a'=\sum_{i \in U_{nor}} v_{t}^{(i)}-\bar{cv}
\end{align}
where $\bar{cv}$ corresponds to the partial result after the adversary has not correctly or partially unlearned. Consequently, the linear homomorphic hash is computed as:
\begin{align}
    \label{eq5-7}
    \mathbf{LHH.Hash}(a')&=\mathbf{LHH.Hash}(\sum_{i \in U_{nor}} v_{t}^{(i)}-\bar{cv})\notag\\
    &=\frac{\prod_{k \in [d]} g_{k}^{\sum_{i \in U_{nor}}v_{t}^{(i)}[k]}}{\prod_{k \in [d]} g_{k}^{\bar{cv}[k]}}
\end{align}
where $d$ represents the vector dimension.

Given our assumption that the forged unlearning result $a'$ of the adversary passes verification, the equality $\mathbf{LHH.Hash}(a')=H$ must hold. Substituting the expressions for $H$ (Eq. (\ref{eq5-5})) and $\mathbf{LHH.Hash}(a')$ (Eq. (\ref{eq5-7})), we derive:
\begin{align}
    \label{eq5-8}
    \prod_{k \in [d]} g_{k}^{\bar{cv}[k]}=\prod_{k \in [d]} g_{k}^{\sum_{j \in U_{unl}}cv_{j}[k]}
\end{align}
However, since $a' \neq a$, it follows from Eq. (\ref{eq5-6}) that $\bar{cv} \neq \sum_{j \in U_{unl}} cv_{j}$. This implies:
\begin{align}
    \label{eq5-9}
    \prod_{k \in [d]} g_{k}^{\bar{cv}[k]} \neq \prod_{k \in [d]} g_{k}^{\sum_{j \in U_{unl}}cv_{j}[k]}
\end{align}
which directly contradicts the derived equality above.

This contradiction disproves our initial assumption. In other words, due to the binding property of the commitment scheme and the collision resistance of the linear homomorphic hash, the probability that a PPT adversary can deceive verifying ISAC devices into accepting a forged unlearning result is negligible. This completes the proof.
\end{proof}

\subsection{Privacy of Input Gradients}
We primarily consider the privacy of ISAC device input gradients in the presence of an untrusted UAV-assisted server and external adversaries. Our protocol ensures that adversaries can only access encrypted gradients and unlearning results, while remaining oblivious to the true input gradients of individual ISAC devices.

\begin{theorem}
In VerFU, no PPT adversary $M$ can infer any true input gradients of ISAC devices from either a single gradient ciphertext or the unlearning result ciphertext.
\end{theorem}

\begin{proof}
Assume the adversary $M$ gains access to: 1) Encrypted gradients or historical contributions $\{x_{i}\}_{u_{i} \in U}$; 2) The encrypted unlearning result $\llbracket a \rrbracket$. The Paillier homomorphic encryption scheme satisfies indistinguishability under chosen-plaintext attacks (IND-CPA), and its security relies on the hardness of the Decisional Composite Residuosity problem, as formally proven in prior work \cite{paillier1999public}. Under the IND-CPA security of the encryption scheme and the hardness of the discrete logarithm problem, the adversary
$M$ cannot distinguish between the real value $x_{i}$ and a simulated value.
Furthermore, the linear homomorphic hash function $\mathbf{LHH.Hash}(\cdot)$ preserves no statistical information about individual gradients $v^{(i)}$ or historical contributions $cv_{i}$ due to its collision resistance (Definition \ref{def1}). Consequently, even with auxiliary knowledge of the unlearning result, the advantage of $M$ in recovering any original input of ISAC devices remains negligible in the security parameter $\kappa$. Overall, the adversary cannot infer the true uploaded content of any individual ISAC device based on its available information. This completes the proof of input gradient privacy in VerFU framework.  
\end{proof}

\section{Performance Evaluation} \label{sec6}

\subsection{Experimental Setup}
\subsubsection{Datasets and Model}
To evaluate the effectiveness of VerFU in LAWN, we adopt image classification tasks as a representative sensing workload, which commonly arises in LAWN, such as vehicle perception, roadside monitoring, and aerial visual inspection. In these scenarios, ISAC devices continuously collect visual data and collaboratively train lightweight perception models at the network edge. Specifically, we conduct experiments on three standard benchmarks, MNIST \cite{deng2012mnist}, Fashion-MNIST \cite{xiao2017fashion}, and EMNIST \cite{cohen2017emnist}, which are widely used in FL studies to assess system behavior under heterogeneous and distributed data settings. For all experiments, we use a convolutional neural network (CNN) based on the LeNet architecture \cite{wortsman2020supermasks}. LeNet is a lightweight model composed of two convolutional layers, one max-pooling layer, and two fully connected layers, making it suitable for deployment on resource-constrained ISAC devices and consistent with the edge intelligence requirements of LAWN.

\subsubsection{Implementation Details}
We simulate an LAWN-based FL environment with 500 ISAC devices, each representing a mobile sensing entity. To reflect the inherent data heterogeneity across devices in LAWN, local datasets are generated using a Dirichlet distribution with parameter $\alpha=1$, resulting in a non-IID data distribution. The FUL process consists of 100 global communication rounds. In each round, 20 ISAC devices are randomly selected to participate in local training and upload encrypted model updates to the UAV-assisted server, capturing the partial participation and dynamic connectivity characteristics of LAWN. Each selected device performs local training for 5 epochs using stochastic gradient descent (SGD) with a learning rate of 0.01. To emulate device mobility and privacy-driven departures in LAWN, we incorporate an unlearning mechanism with unlearning rates of {0\%,10\%,20\%,40\%}. For example, under a 20\% unlearning rate, five ISAC devices are randomly selected every five rounds to request the removal of their historical contributions, modeling scenarios where devices leave the coverage region or request unlearning proactively.

\begin{figure*}[tbp]
\centering
\begin{subfigure}[b]{0.31\textwidth}
    \centering
    \includegraphics[width=\linewidth]{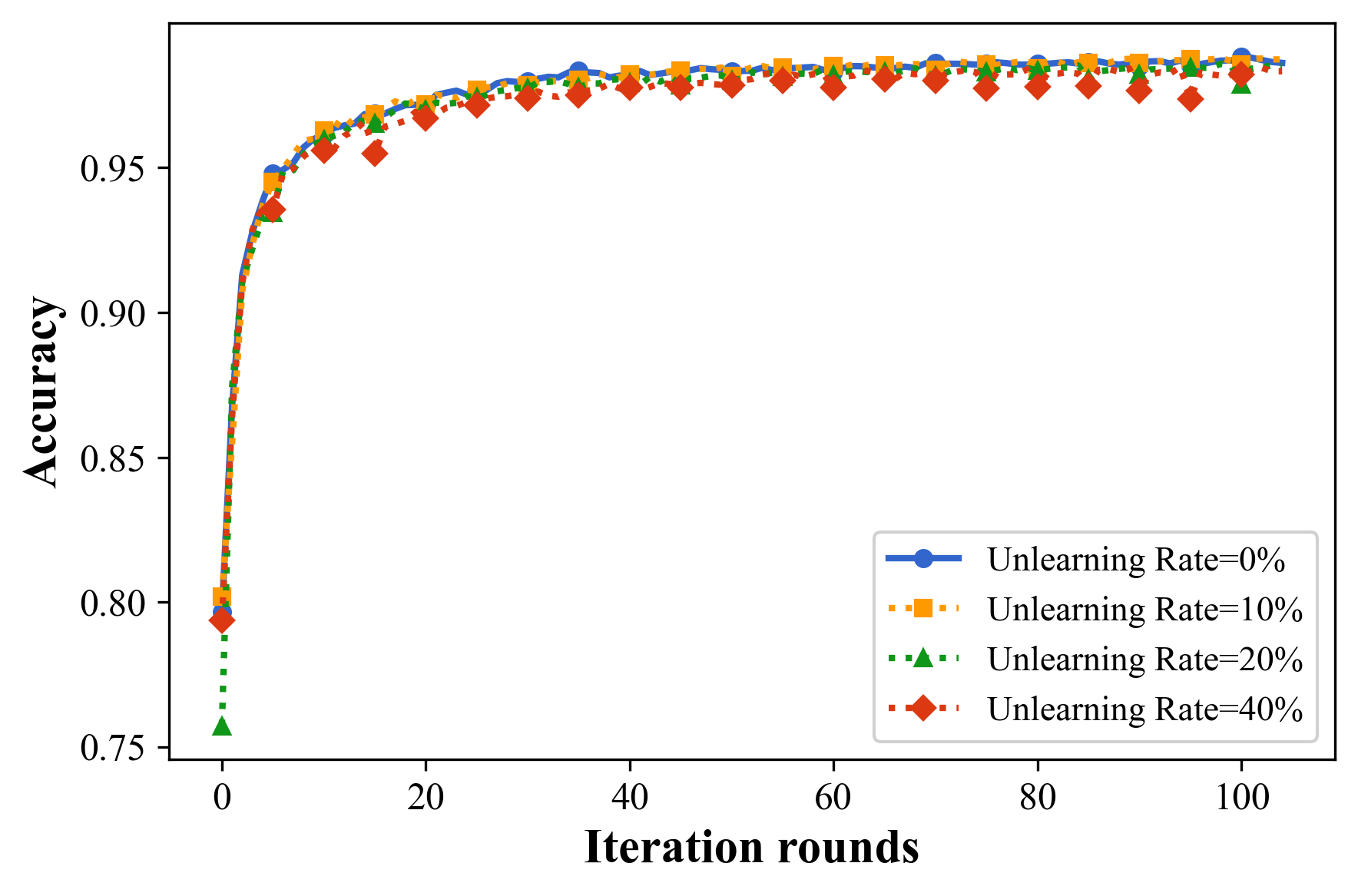}
    \caption{MNIST}
\end{subfigure}
\begin{subfigure}[b]{0.31\textwidth}
    \centering
    \includegraphics[width=\linewidth]{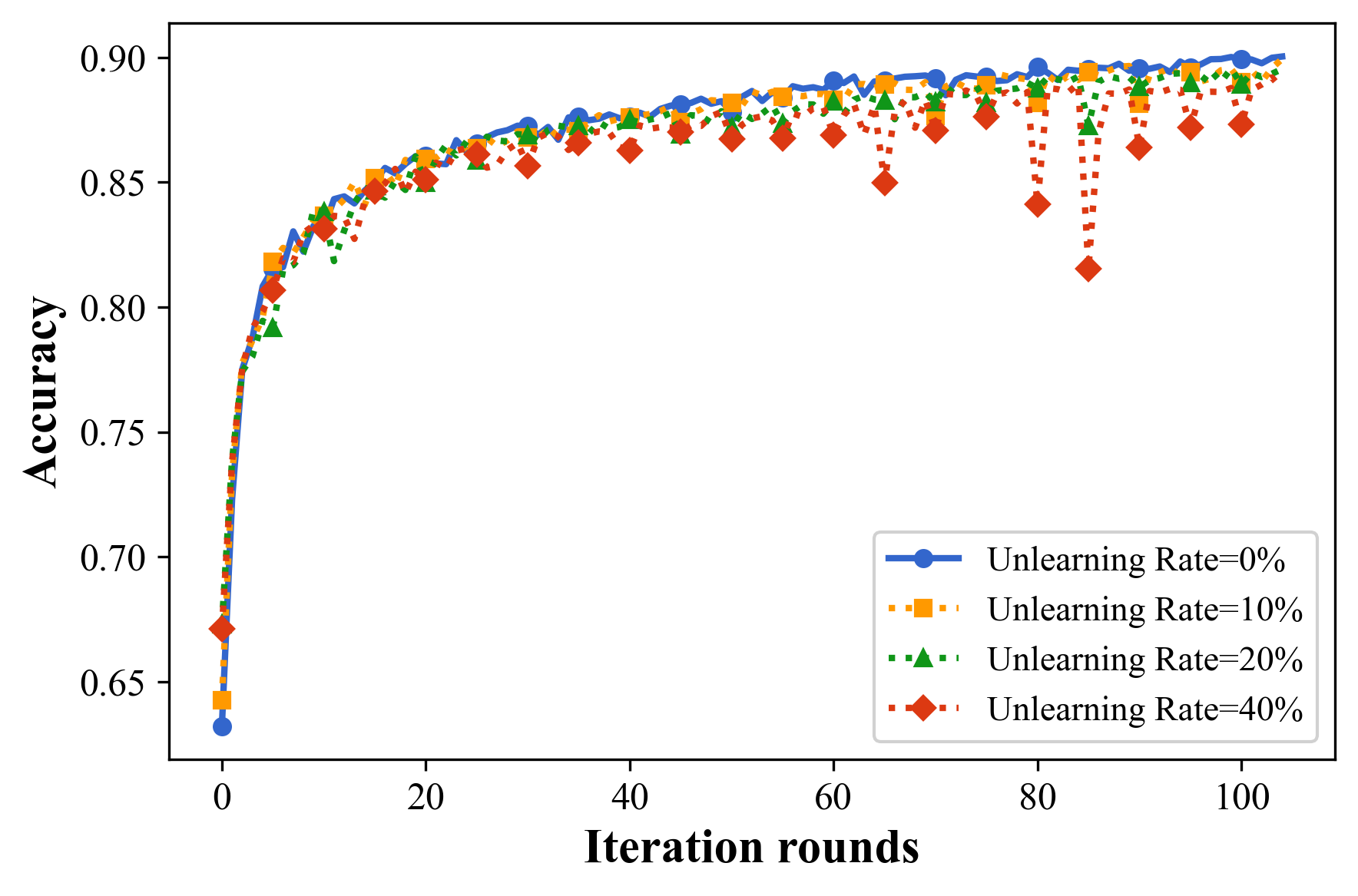}
    \caption{Fashion-MNIST}
\end{subfigure}
\begin{subfigure}[b]{0.31\textwidth}
    \centering
    \includegraphics[width=\linewidth]{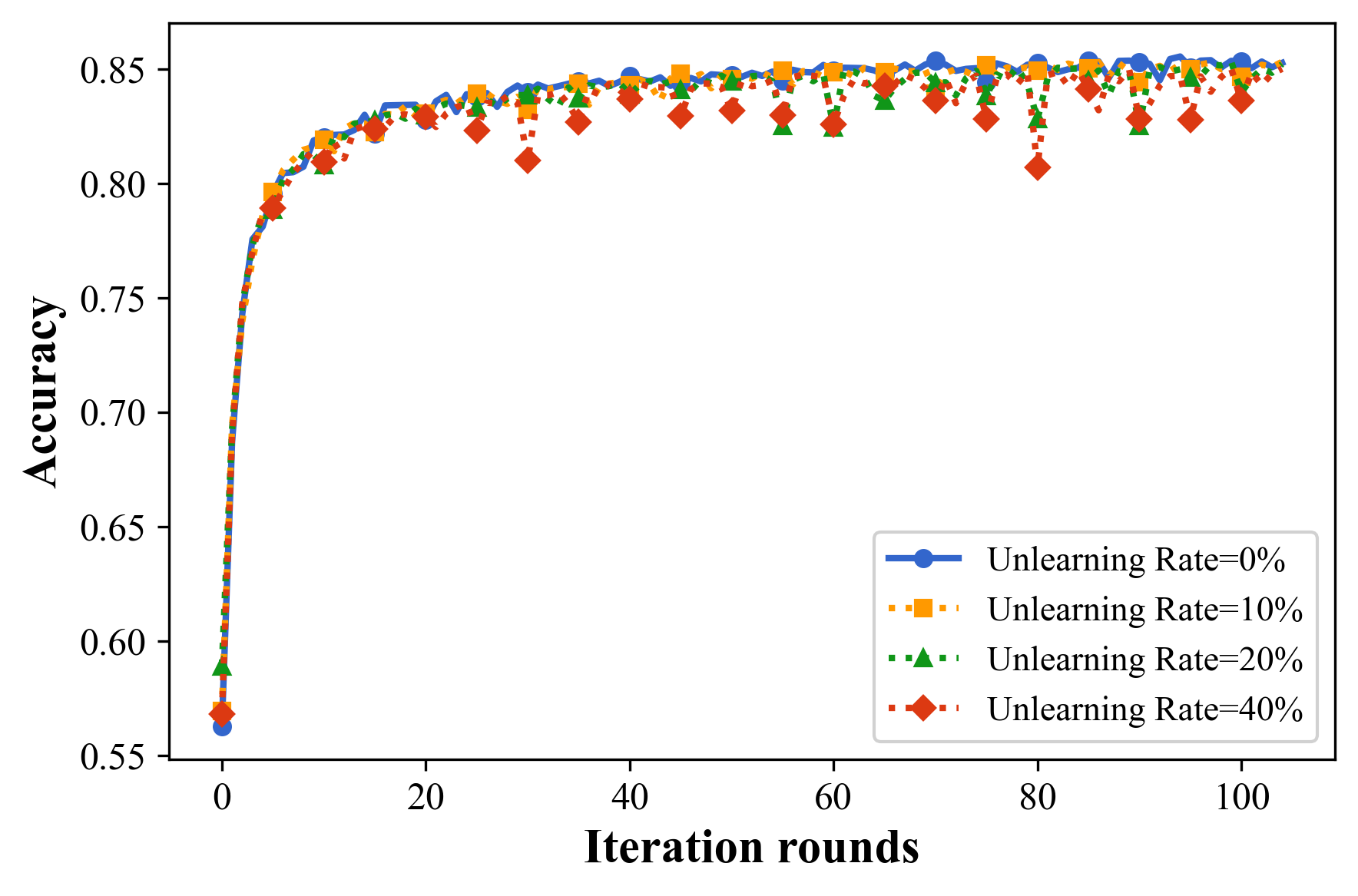}
    \caption{EMNIST}
\end{subfigure}
\caption{Accuracy performance under dynamic ISAC device unlearning in VerFU.}
\label{fig:acc}
\end{figure*}
\begin{figure*}[tbp]
\centering
\begin{subfigure}[b]{0.31\textwidth}
    \centering
    \includegraphics[width=\linewidth]{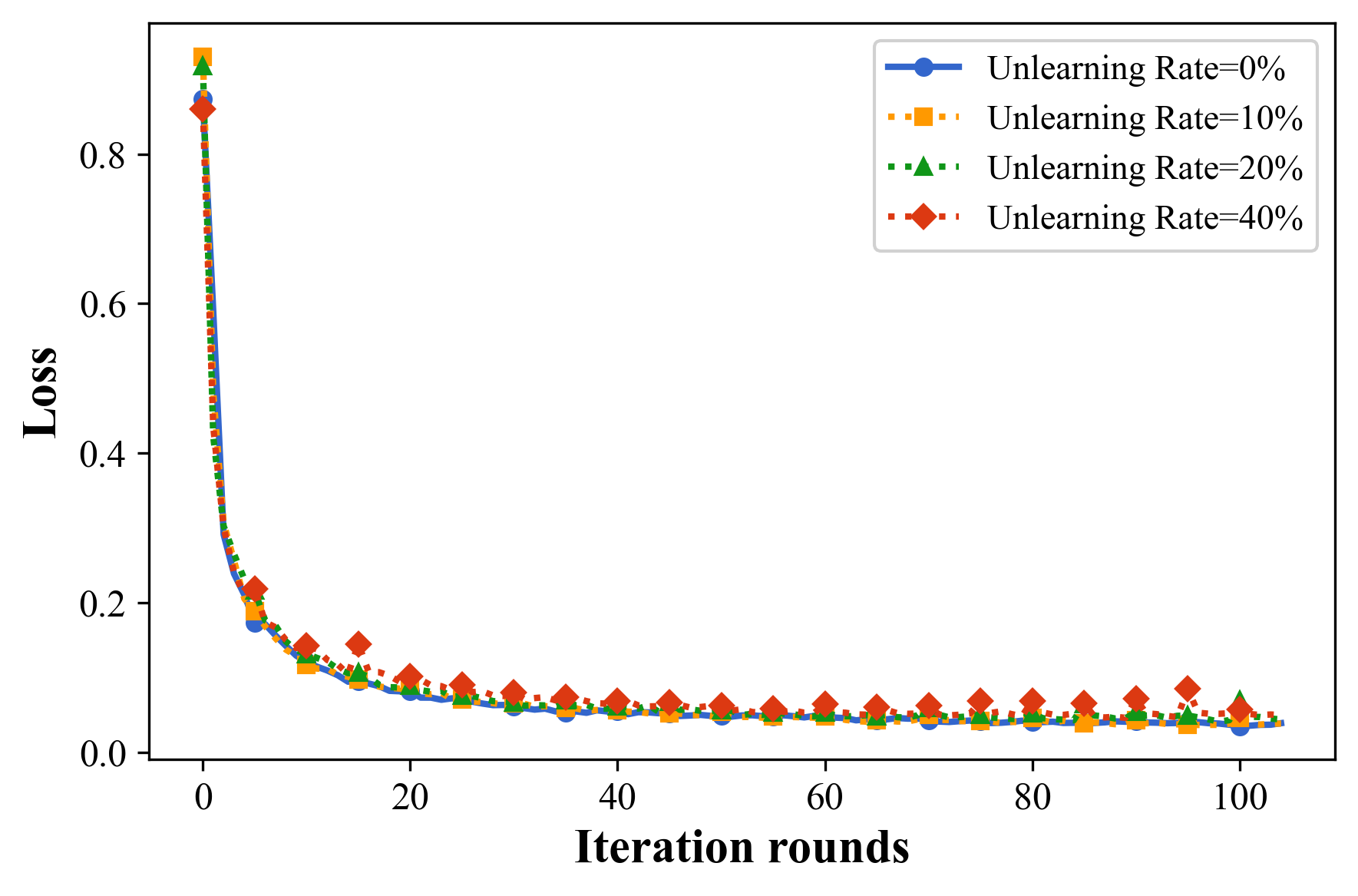}
    \caption{MNIST}
\end{subfigure}
\begin{subfigure}[b]{0.31\textwidth}
    \centering
    \includegraphics[width=\linewidth]{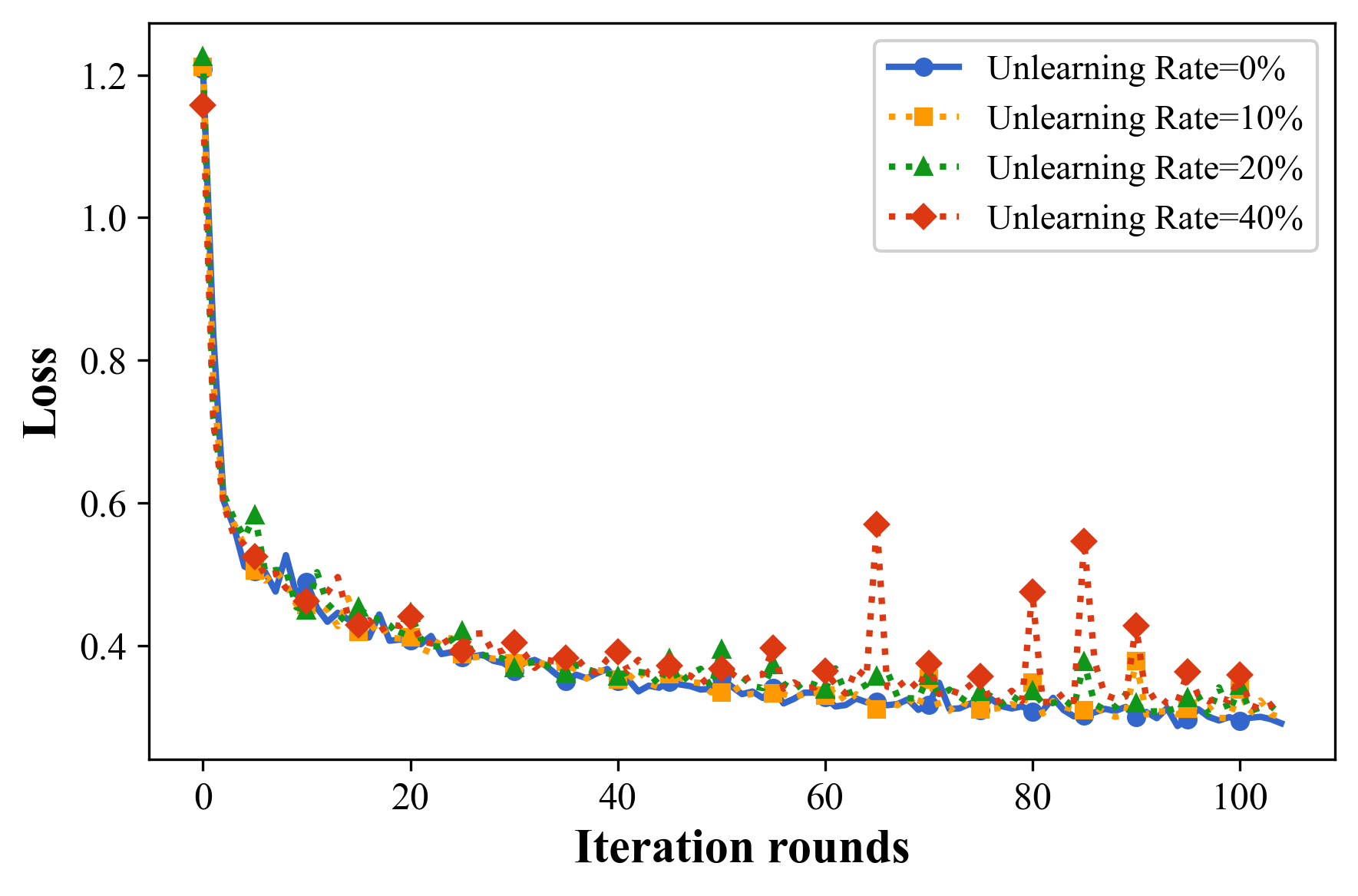}
    \caption{Fashion-MNIST}
\end{subfigure}
\begin{subfigure}[b]{0.31\textwidth}
    \centering
    \includegraphics[width=\linewidth]{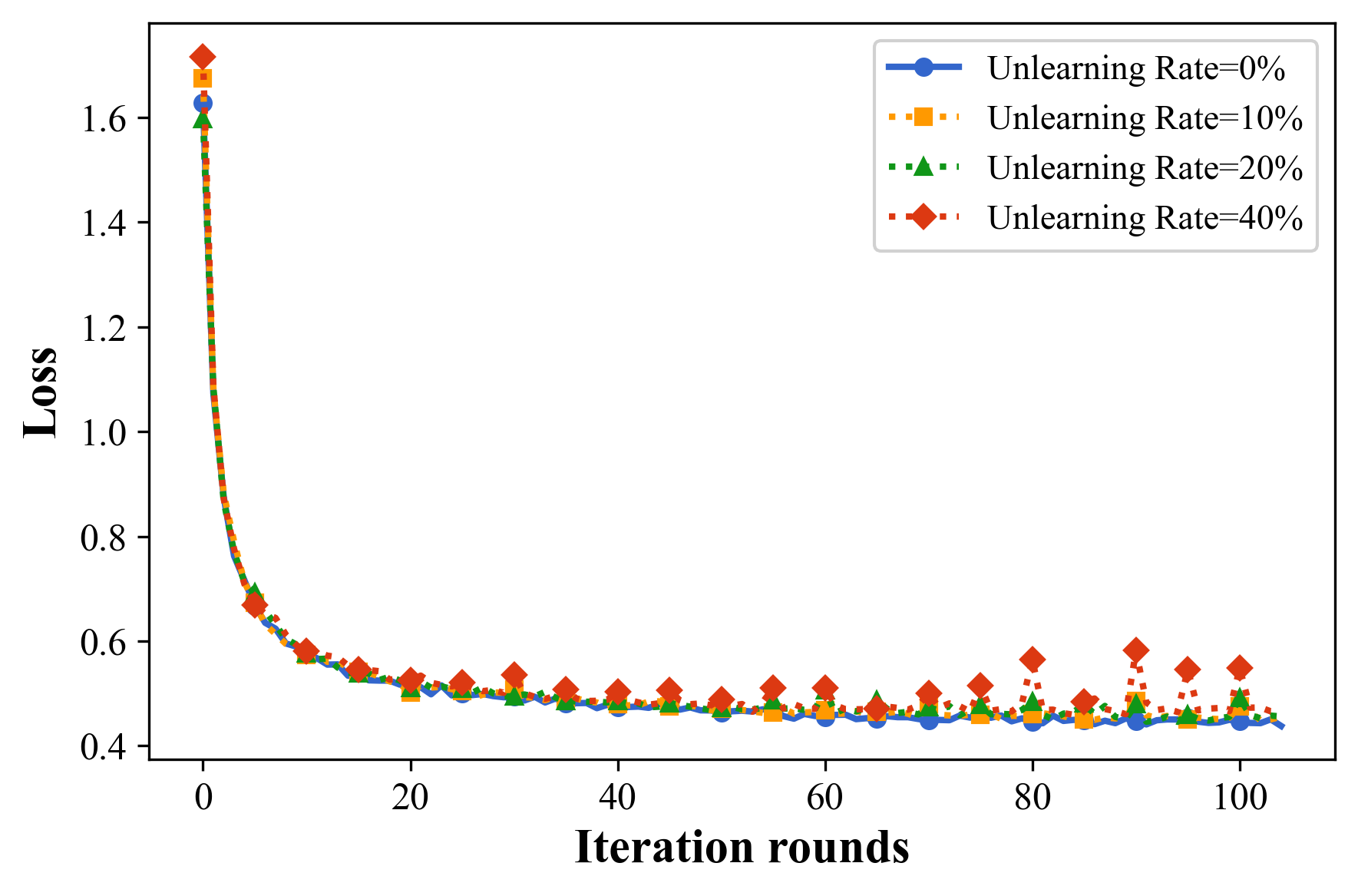}
    \caption{EMNIST}
\end{subfigure}
\caption{Loss comparison under dynamic ISAC device unlearning in VerFU.}
\label{fig:loss}
\end{figure*}
\begin{figure*}[tbp]
\centering
\begin{subfigure}[b]{0.31\textwidth}
    \centering
    \includegraphics[width=\linewidth]{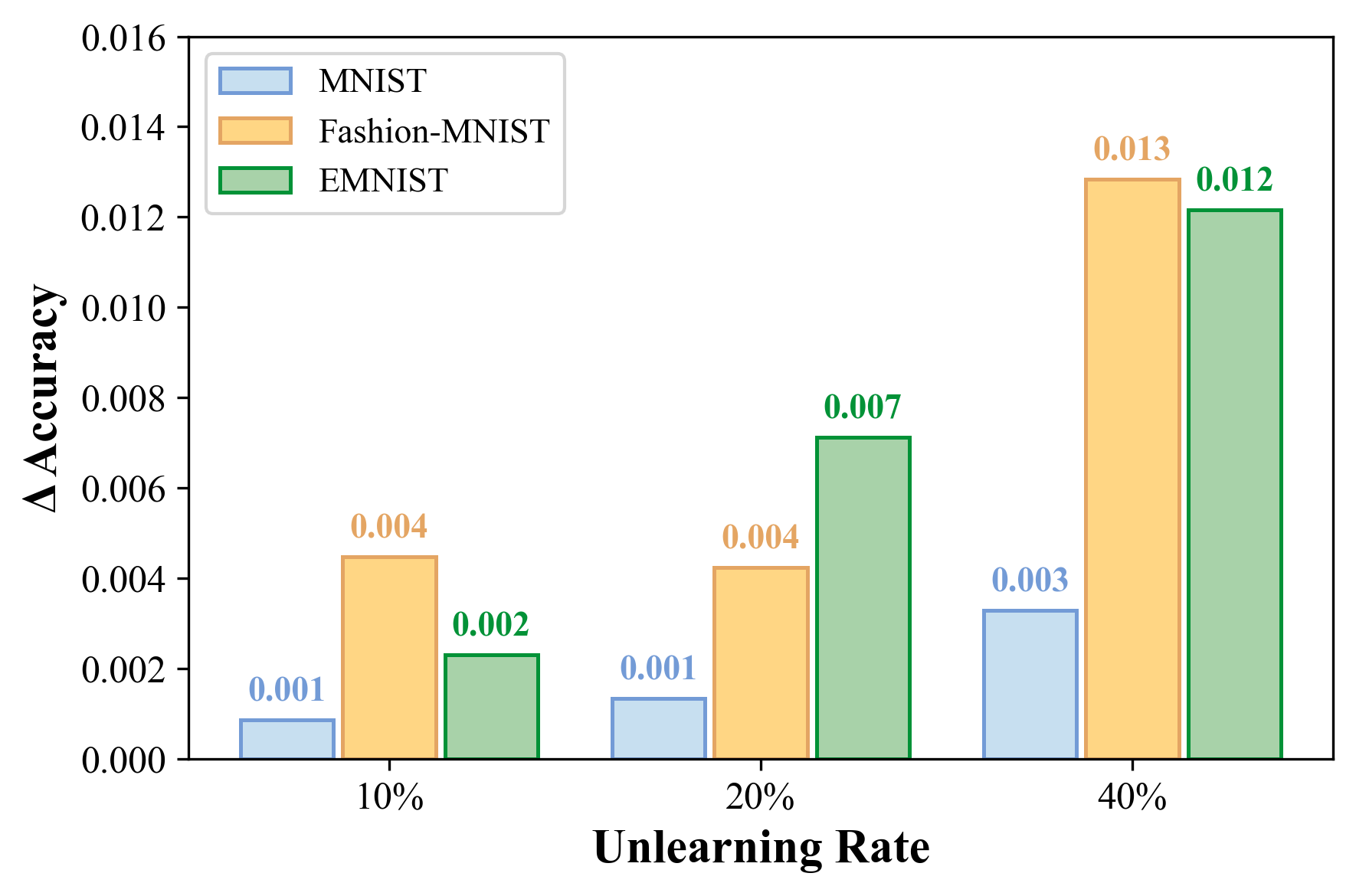}
    \caption{$\Delta$ Accuracy}
\end{subfigure}
\begin{subfigure}[b]{0.31\textwidth}
    \centering
    \includegraphics[width=\linewidth]{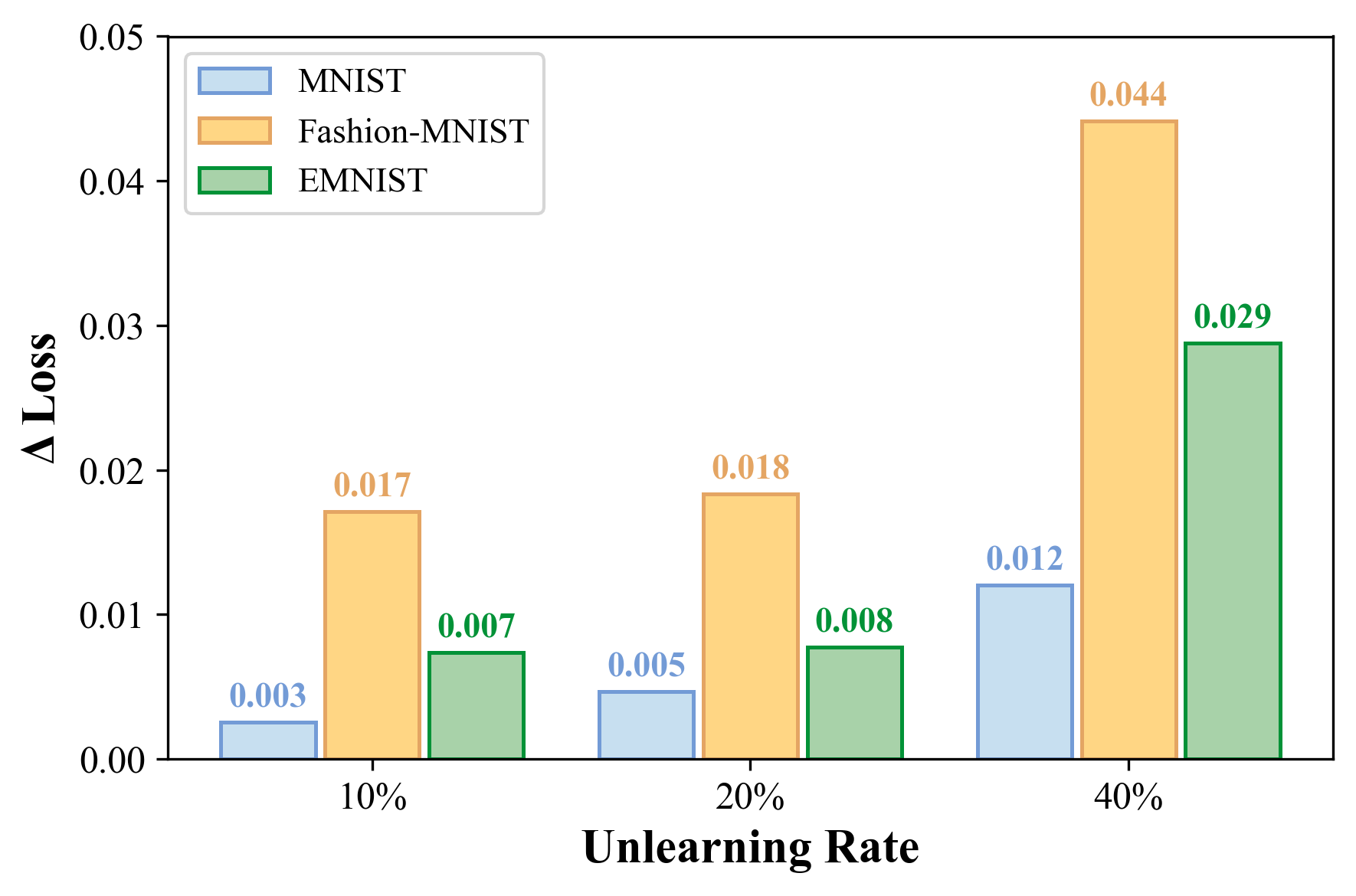}
    \caption{$\Delta$ Loss}
\end{subfigure}
\begin{subfigure}[b]{0.31\textwidth}
    \centering
    \includegraphics[width=\linewidth]{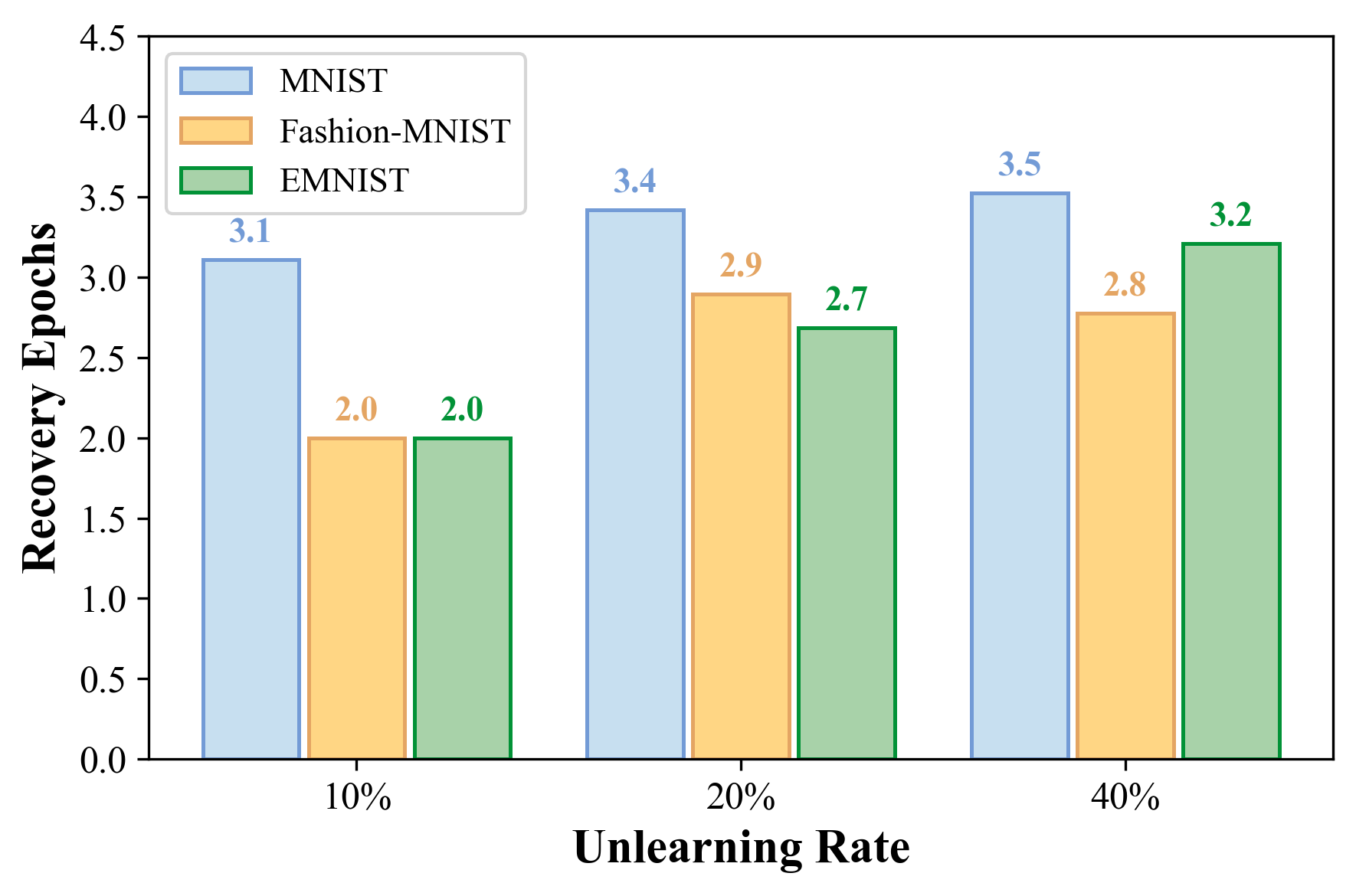}
    \caption{Accuracy Recovery Epochs}
\end{subfigure}
\caption{Impact of dynamic ISAC device unlearning on model utility in VerFU.}
\label{fig:model_utility}
\end{figure*}

\subsubsection{Baseline Methods}
For a comparative analysis of overhead, three representative verification approaches were selected as baselines: Proof of Learning (PoL) \cite{jia2021proof}, its privacy-preserving extension PoLHE, and zk-Proof of Training (zkPoT) \cite{garg2023experimenting}. The PoL framework authenticates training integrity by verifying intermediate model states and data batches through the inherent stochasticity of stochastic gradient descent (SGD). The PoL-HE variant enhances this architecture through homomorphic encryption to preserve data privacy while enabling the verifiable reproducibility of designated training phases. In contrast, zkPoT implements a zero-knowledge verification paradigm that combines MPC-in-the-head techniques with zk-SNARKs to generate proofs for efficient compliance auditing of training procedures. We introduce these established validation computation mechanisms into FUL to provide a basic baseline for evaluating the proposed VerFU framework designed for LAWN.

\subsection{Post-unlearning Model Utility}
To assess the impact of unlearning on collaborative model utility in LAWN, experiments were conducted on three datasets to evaluate changes in accuracy and loss under different unlearning rates. This evaluation reflects the stability of shared perception models when ISAC devices dynamically leave the learning process due to mobility or privacy requirements. As shown in Fig. \ref{fig:acc}, the accuracy curves for four different unlearning rates consistently reach high levels after 100 FL global rounds. Even with a 40\% unlearning rate, VerFU achieves accuracies of 98.33\%, 88.69\%, and 84.86\% on MNIST, Fashion-MNIST, and EMNIST, respectively. Furthermore, as the unlearning rate increases, the global model exhibits a slight decline in accuracy. For instance, the accuracy on MNIST under unlearning rates of 0\%, 10\%, 20\%, and 40\% is 98.80\%, 98.76\%, 98.42\%, and 98.33\%, respectively. Similarly, the loss values of VerFU were measured across the three datasets under varying unlearning rates, as illustrated in Fig. \ref{fig:loss}. The results indicate that unlearning does not lead to divergence, and the global model remains capable of recovering and converging as expected.

To further investigate unlearning-induced performance fluctuations, we quantitatively analyze three key metrics: accuracy drop ($\Delta acc= \lvert acc_{pre}-acc_{post} \rvert$), loss increase ($\Delta loss= \lvert loss_{pre}-loss_{post} \rvert$), and the number of rounds required for accuracy recovery. As depicted in Figs. \ref{fig:model_utility}(a) and \ref{fig:model_utility}(b), increasing the unlearning rate results in a performance decline, yet the maximum accuracy drop is only 1.3\%, while the highest observed loss increase is merely 0.044. Fig. \ref{fig:model_utility}(c) further demonstrates that despite transient performance fluctuations caused by unlearning requests from a subset of ISAC devices, the model recovers within 2–4 rounds of normal training. Overall, these results indicate that VerFU preserves robust post-unlearning model utility under dynamic participation. Even when some ISAC devices request unlearning, the global model remains stable and quickly adapts, which is critical for maintaining reliable edge intelligence services in LAWN.

\subsection{Communication Overhead}
Communication efficiency is a critical requirement for LAWN deployments, where UAV-assisted servers and ISAC devices operate under limited bandwidth and dynamic wireless conditions. To evaluate the communication efficiency of VerFU, Table \ref{tab:commu_stage} presents a breakdown of communication overhead at different phases. It is worth noting that the communication cost caused by the verification of unlearning is negligible. Even as the unlearning rate increases from 0\% to 40\%, the amount of data transmitted between each ISAC device and the UAV-assisted server remains below 1KB, indicating that verification imposes low burden on low-altitude wireless links. Furthermore, we compare the communication overhead introduced by VerFU and baseline methods across different datasets, as summarized in Table \ref{tab:communication}. We define Light VerFU as a variant of VerFU that does not incorporate HE and only includes the essential components for verifiable unlearning. Under this setting, the communication overhead remains as low as approximately 6.20MB across all datasets, highlighting its lightweight design. In contrast, the baseline method PoL incurs 12.41MB of overhead on both MNIST and Fashion-MNIST, and 12.46MB on EMNIST. When enhanced with HE for stronger privacy protection, VerFU introduces around 74MB of communication cost, which is still significantly lower than that of PoLHE and zkPoT. 
\begin{table}[tbp]
    \renewcommand{\arraystretch}{1.2}
    \centering
    \caption{Comparison of communication overheads for VerFU and baseline methods}\label{tab:communication}
    \begin{tabular}{cccc} 
    \toprule
     & \multicolumn{3}{c}{Communication Overhead} \\
    \cmidrule(l){2-4}
    Method/Dataset & MNIST & Fashion-MNIST & EMNIST\\ 
    \midrule
    \textbf{Light VerFU} & 6.20 MB & 6.20 MB & 6.23 MB\\
    PoL & 12.41 MB & 12.41 MB & 12.46 MB\\
    \textbf{VerFU} & 74.42 MB & 74.42 MB & 74.72 MB\\
    PoLHE & 853.82 MB & 853.82 MB & 857.43 MB\\ 
    zkPoT & 147.45 MB & 147.45 MB & 258.78 MB \\
    \bottomrule
    \end{tabular}
\end{table}
Expressly, the overhead of PoLHE on all datasets is less than 860MB, while zkPoT results in 147.45MB, 147.45MB, and 258.78MB of overhead on MNIST, Fashion-MNIST, and EMNIST, respectively. Although zkPoT outperforms PoLHE in communication efficiency, its reliance on zero-knowledge proofs still results in notable overhead. These results demonstrate that VerFU achieves a balanced trade-off between verification integrity and communication efficiency. Even when employing HE-based privacy protection, the reduced overhead makes VerFU a scalable and practical solution for LAWN applications that require efficient and verifiable FUL.

\begin{table*}[htbp]
\centering
\begin{threeparttable}
\caption{Communication Overhead of VerFU at Different Phases\tnote{1}}
\label{tab:commu_stage}
\renewcommand{\arraystretch}{1.2}
\setlength{\tabcolsep}{2.5pt}
\begin{tabular}{cccccc|cccc|cccc}
\toprule
 \multirow{2}{*}{Unl. R.} & & \multicolumn{4}{c|}{MNIST} & \multicolumn{4}{c|}{Fashion-MNIST} & \multicolumn{4}{c}{EMNIST} \\
\cmidrule{3-14}
 & & Pre. & Agg. \& Unl. & Ver. & Total & Pre. & Agg. \& Unl. & Ver. & Total & Pre. & Agg. \& Unl. & Ver. & Total \\
\midrule
 \multirow{2}{*}{0\%} & Device & \textbf{0.00}KB & 74.42MB & \textbf{0.00}KB & 74.42MB & \textbf{0.00}KB & 74.42MB & \textbf{0.00}KB & 74.42MB & \textbf{0.00}KB & 74.72MB & \textbf{0.00}KB & 74.72MB \\
 & Server & \textbf{0.00}KB & 74.42MB & \textbf{0.00}KB & 74.42MB & \textbf{0.00}KB & 74.42MB & \textbf{0.00}KB & 74.42MB & \textbf{0.00}KB & 74.72MB & \textbf{0.00}KB & 74.72MB \\
\cline{2-14}
 \multirow{2}{*}{10\%} & Device & \textbf{0.03}KB & 74.42MB & \textbf{0.03}KB & 74.42MB & \textbf{0.03}KB & 74.42MB & \textbf{0.03}KB & 74.42MB & \textbf{0.03}KB & 74.72MB & \textbf{0.03}KB & 74.72MB \\
 & Server & \textbf{0.62}KB & 74.42MB & \textbf{0.16}KB & 74.42MB & \textbf{0.62}KB & 74.42MB & \textbf{0.16}KB & 74.42MB & \textbf{0.62}KB & 74.72MB & \textbf{0.16}KB & 74.72MB \\
\cline{2-14}
 \multirow{2}{*}{20\%} & Device & \textbf{0.03}KB & 74.42MB & \textbf{0.03}KB & 74.42MB & \textbf{0.03}KB & 74.42MB & \textbf{0.03}KB & 74.42MB & \textbf{0.03}KB & 74.72MB & \textbf{0.03}KB & 74.72MB \\
 & Server & \textbf{0.62}KB & 74.42MB & \textbf{0.16}KB & 74.42MB & \textbf{0.62}KB & 74.42MB & \textbf{0.16}KB & 74.42MB & \textbf{0.62}KB & 74.72MB & \textbf{0.16}KB & 74.72MB \\
\cline{2-14}
 \multirow{2}{*}{40\%} & Device & \textbf{0.03}KB & 74.42MB & \textbf{0.03}KB & 74.42MB & \textbf{0.03}KB & 74.42MB & \textbf{0.03}KB & 74.42MB & \textbf{0.03}KB & 74.72MB & \textbf{0.03}KB & 74.72MB \\
 & Server & \textbf{0.62}KB & 74.42MB & \textbf{0.31}KB & 74.42MB & \textbf{0.62}KB & 74.42MB & \textbf{0.31}KB & 74.42MB & \textbf{0.62}KB & 74.72MB & \textbf{0.31}KB & 74.72MB \\
\bottomrule
\end{tabular}
\begin{tablenotes}
\footnotesize
\item[1] The total number of ISAC devices is set to 500, and the FUL process runs for 100 rounds, with 20 ISAC devices selected per round. The unlearning rate indicates the proportion of unlearning ISAC devices across all rounds relative to the total number of ISAC devices. Pre., Agg. \& Unl., and Ver. denote the preparation phase, aggregation and unlearning phase, and verification phase, respectively. The \textbf{Bold} numbers mark the verification communication overhead introduced by VerFU.
\end{tablenotes}
\end{threeparttable}
\end{table*}

\subsection{Verification Overhead}
\begin{table*}[htbp]
\centering
\begin{threeparttable}
\caption{Computation Overhead of VerFU at Different Phases\tnote{1}}
\label{tab:compu_stage}
\renewcommand{\arraystretch}{1.2}
\setlength{\tabcolsep}{1.3pt}
\begin{tabular}{cccccc|cccc|cccc}
\toprule
 \multirow{2}{*}{Unl. R.} & & \multicolumn{4}{c|}{MNIST} & \multicolumn{4}{c|}{Fashion-MNIST} & \multicolumn{4}{c}{EMNIST} \\
\cmidrule{3-14}
 & & Pre. & Agg. \& Unl. & Ver. & Total & Pre. & Agg. \& Unl. & Ver. & Total & Pre. & Agg. \& Unl. & Ver. & Total \\
\midrule
 \multirow{2}{*}{0\%} & Device & 0.18s & 1508.00s & 452.56s & 1960.75s &  0.16s & 1553.36s & 465.66s & 2019.18s & 2.07s & 1671.38s & 494.84s & 2168.29s \\
 & Server & 0.00s & 0.01s & 0.00s & 0.01s & 0.00s & 0.01s & 0.00s & 0.01s & 0.00s & 0.01s & 0.00s & 0.01s \\
\cline{2-14}
 \multirow{2}{*}{10\%} & Device & 12.24s & 1590.07s & 475.61+\textbf{11.53}s & 2089.45s & 12.67s & 1602.08s & 481.39+\textbf{11.70}s & 2107.91s & 14.27s & 1523.91s & 482.90+\textbf{12.06}s & 2033.14s \\
 & Server & 0.00s & 0.02s & 0.00s & 0.02s & 0.00s & 0.02s & 0.00s & 0.02s & 0.00s & 0.02s & 0.00s & 0.02s \\
\cline{2-14}
 \multirow{2}{*}{20\%} & Device & 12.86s & 1581.96s & 499.03+\textbf{11.69}s & 2105.54s & 12.60s & 1516.13s & 480.01+\textbf{11.88}s & 2020.62s & 13.58s & 1567.84s & 468.29+\textbf{12.95}s & 2062.66s \\
 & Server & 0.00s & 0.02s & 0.00s & 0.02s & 0.00s & 0.03s & 0.00s & 0.03s & 0.00s & 0.02s & 0.00s & 0.02s \\
\cline{2-14}
 \multirow{2}{*}{40\%} & Device & 12.34s & 1672.54s & 510.79+\textbf{12.18}s & 2207.85s & 12.34s & 1638.29s & 540.66+\textbf{13.51}s & 2204.80 & 12.72s & 1527.85s & 474.76+\textbf{14.30}s & 2029.63s \\
 & Server & 0.00s & 0.02s & 0.00s & 0.02s & 0.00s & 0.02s & 0.00s & 0.02s & 0.00s & 0.02s & 0.00s & 0.02s \\
\bottomrule
\end{tabular}
\begin{tablenotes}
\footnotesize
\item[1] The total number of ISAC devices is set to 500, and the FUL process runs for 100 rounds, with 20 ISAC devices selected per round. The unlearning rate indicates the proportion of unlearning ISAC devices across all rounds relative to the total number of ISAC devices. Pre., Agg. \& Unl., and Ver. denote the preparation phase, aggregation and unlearning phase, and verification phase, respectively. The \textbf{Bold} numbers mark the verification computation cost introduced by VerFU.
\end{tablenotes}
\end{threeparttable}
\end{table*}

Efficient verification is essential to ensure that federated unlearning can be executed on demand without becoming a computational bottleneck in LAWN. Table \ref{tab:compu_stage} provides a detailed analysis of the computation overhead incurred by VerFU across different phases and unlearning rates on various datasets. As shown, the majority of computation time is attributed to encryption and decryption operations associated with HE. In contrast, the additional cost introduced by the verification phase remains minimal. For instance, even when the unlearning rate reaches 40\%, the verification overhead consistently remains below 15s across all datasets. Building on this analysis, Table \ref{tab:verify} further compares the verification efficiency between VerFU and baseline methods on the three benchmark datasets. VerFU achieves the lowest verification latency, requiring only 11.53s, 11.70s, and 12.06s on MNIST, Fashion-MNIST, and EMNIST, respectively. These results demonstrate the high efficiency of VerFU in enabling verifiable federated unlearning with negligible additional verification overhead.
\begin{table}[tbp]
    \renewcommand{\arraystretch}{1.2}
    \centering
    \caption{Comparison of verification overheads for VerFU and baseline methods}\label{tab:verify}
    \begin{tabular}{cccc} 
    \toprule
     & \multicolumn{3}{c}{Verification Overhead} \\
    \cmidrule(l){2-4}
    Method/Dataset & MNIST & Fashion-MNIST & EMNIST\\ 
    \midrule
    \textbf{VerFU} & 11.53s & 11.70s & 12.06s \\
    PoL & 104.38s & 105.93s & 1137.56s \\
    PoLHE & 17415.82s & 17808.72s & 191565.13s \\ 
    zkPoT & 1831.21s & 1912.27s & 24746.54s \\
    \bottomrule
    \end{tabular}
\end{table}
However, PoL produces significantly higher verification costs, taking approximately 104.38s on MNIST, 105.93s on Fashion-MNIST, and up to 1137.56s on EMNIST. While zkPoT achieves an efficiency improvement over PoLHE, the computational burden of zero-knowledge proof generation and verification still results in verification times of 1831.21s, 1912.27s, and 24746.54s on the respective datasets. To conclude, VerFU demonstrates superior verification efficiency across all baselines. By keeping verification latency low, VerFU empowers ISAC devices with practical verifiability guarantees while remaining compatible with the real-time and resource-constrained nature of LAWN.

\section{Related Works} \label{sec7}
\subsection{Federated unlearning}
Machine unlearning \cite{cao2015towards, li2025machine} requires the elimination of the influence of special training samples on a model. The most intuitive approach is retraining from scratch on the updated training dataset after removing the specified subset. However, this method results in substantial retraining overhead. To mitigate this overhead, Bourtoule et al. \cite{bourtoule2021machine} proposed a machine unlearning scheme based on slicing techniques, where the training data is partitioned into multiple disjoint shards, and each shard independently trains its model while storing the model parameter states. Unlearning can then be achieved by retraining only the slice model that contains the deleted data. Ginart et al. \cite{ginart2019making} adopted a divide-and-conquer k-means algorithm, partitioning data through a tree structure and recalculating only impacted subproblems. Additionally, existing works \cite{golatkar2020eternal, mehta2022deep} proposed fine-tuning the global model using the remaining training dataset after deletion to reduce the impact of the removed data.

Due to privacy constraints, federated unlearning cannot access local datasets. Existing machine unlearning methods, which depend on explicit access to both deleted and remaining data, are difficult to extend straightforwardly to federated unlearning. Liu et al. \cite{liu2021federaser} pioneered a novel approach for data unlearning in federated learning by leveraging the client parameter updates stored on the server to iteratively calibrate the global model from previous rounds before the unlearning request. This method eliminates the influence of specific client data from the global model without requiring retraining. In addition to methods that calibrate historical updates \cite{liu2021federaser, yuan2023federated, fu2024client}, several works have explored federated unlearning from perspectives such as gradient subtraction \cite{guo2023fast, zhang2023fedrecovery, jiang2025towards}, gradient scaling \cite{gao2024verifi}, or gradient ascent training \cite{halimi2022federated, li2023subspace, wang2024server}. For instance, Zhang et al. \cite{zhang2023fedrecovery} introduced a differential privacy-based method that eliminates client impact by subtracting gradient residuals and applying the Gaussian mechanism to maintain statistical indistinguishability. Wang et al. \cite{wang2024server} devised a scheme to remove the impact of low-quality data on the global model in federated learning. By performing gradient ascent training on the target client and applying boost training to mitigate bias, the scheme efficiently eliminates the influence of low-quality data on the global model.

\subsection{Aggregation and unlearning verification}
For aggregation verification, researchers primarily employ cryptographic techniques to verify the server aggregation process and results. Xu et al. \cite{xu2019verifynet} were the first to propose a verifiable secure aggregation scheme based on double masking and homomorphic hash functions. Since each gradient vector during transmission is encrypted using homomorphic hash functions and a pseudorandom function, this scheme incurs significant additional verification overhead. Specifically, each dimension of the vector requires four pairings and one exponentiation in the cyclic group to verify the aggregation result. Zheng et al. \cite{zheng2022aggregation} proposed a hardware-assisted trusted execution environment that verifies the aggregation results by signing the computation process and output, thereby ensuring the computational integrity of the server. Wang et al. \cite{wang2022vosa} introduced a verifiable and oblivious secure aggregation protocol, where users verify the correctness of the aggregation results through designed bilinear pairings, effectively preventing dishonest behavior of the aggregation server and ensuring the integrity and verifiability of the aggregation results.
To reduce verification overhead, Fu et al. \cite{fu2020vfl} designed a verifiable federated learning framework that combines Lagrange interpolation and the Chinese Remainder Theorem. By carefully setting interpolation points, the correctness of the aggregated gradients can be verified, ensuring that participants can independently detect even if the aggregation server maliciously returns forged gradients. They demonstrated that the overhead of this verification method does not increase with the number of participants.
Hahn et al. \cite{hahn2021versa} designed a verifiable secure aggregation scheme for cross-device federated learning. By replacing computationally expensive bilinear pairing operations with the lightweight pseudorandom generator, the scheme achieves efficient double aggregation with low verification cost for both server and clients.

For unlearning verification, Weng et al. \cite{weng2024proof} utilized trusted execution environments like the method in \cite{zheng2022aggregation} to enforce verifiability, but their reliance on trusted hardware limits practicality. Inspired by backdoor attacks, Sommer et al. \cite{sommer2020towards} proposed a verification mechanism for machine unlearning, which allows users to embed unique backdoor triggers into the data and verify data deletion by testing model responses. Nevertheless, this mechanism is vulnerable to server detection of the backdoor triggers and incurs high retraining overhead. To overcome these issues, Guo et al. \cite{guo2023verifying} introduced invisible backdoor markers and incremental learning to prevent forgery and improve retraining efficiency. Gao et al. \cite{gao2024verifi} devised a unified framework for verifiable federated unlearning. This framework allows departing participants to select specific sample data for marking and maintain low loss after model fine-tuning. After the server performs the unlearning operation, the global model is checked to determine whether it can restore high loss on the marked set, thus confirming the effectiveness of unlearning.

\section{Conclusion} \label{sec8}
In this paper, we propose a privacy-preserving and client-verifiable federated unlearning framework, called VerFU, which is designed for LAWN. It allows ISAC devices to verify the execution of unlearning requests by the UAV-assisted server without having to access original data samples. Specifically, VerFU leverages linear homomorphic hash, commitment schemes, and homomorphic encryption to construct tamper-proof contribution records, enabling secure aggregation and unlearning in the encrypted domain. In addition, VerFU supports concurrent unlearning and verification by incorporating a lightweight client state tagging mechanism and a consistency check over hashed linear combinations. We provide a rigorous security analysis of the framework and evaluate VerFU extensively on three representative datasets. Experimental evaluations show that VerFU effectively recovers model utility after unlearning while incurring low verification and communication overhead. These results indicate that VerFU is a practical and scalable solution for enforcing verifiable federated unlearning in LAWN, contributing to privacy-preserving and efficient edge intelligence in dynamic low-altitude environments.

\bibliographystyle{IEEEtran}
\bibliography{ref}

\begin{thebibliography}{10}
\providecommand{\url}[1]{#1}
\csname url@samestyle\endcsname
\providecommand{\newblock}{\relax}
\providecommand{\bibinfo}[2]{#2}
\providecommand{\BIBentrySTDinterwordspacing}{\spaceskip=0pt\relax}
\providecommand{\BIBentryALTinterwordstretchfactor}{4}
\providecommand{\BIBentryALTinterwordspacing}{\spaceskip=\fontdimen2\font plus
\BIBentryALTinterwordstretchfactor\fontdimen3\font minus \fontdimen4\font\relax}
\providecommand{\BIBforeignlanguage}[2]{{%
\expandafter\ifx\csname l@#1\endcsname\relax
\typeout{** WARNING: IEEEtran.bst: No hyphenation pattern has been}%
\typeout{** loaded for the language `#1'. Using the pattern for}%
\typeout{** the default language instead.}%
\else
\language=\csname l@#1\endcsname
\fi
#2}}
\providecommand{\BIBdecl}{\relax}
\BIBdecl

\bibitem{yuan2025ground}
W.~Yuan, Y.~Cui, J.~Wang, F.~Liu, G.~Sun, T.~Xiang, J.~Xu, S.~Jin, D.~Niyato, S.~Coleri \emph{et~al.}, ``From ground to sky: Architectures, applications, and challenges shaping low-altitude wireless networks,'' \emph{arXiv preprint arXiv:2506.12308}, 2025.

\bibitem{luo2025toward}
G.~Luo, J.~Li, Q.~Zhang, Z.~Feng, Q.~Yuan, Y.~Lin, H.~Zhang, N.~Cheng, and P.~Zhang, ``Toward low-altitude airspace management and uav operations: Requirements, architecture and enabling technologies,'' \emph{IEEE Wireless Communications}, 2025.

\bibitem{wang2025toward}
Y.~Wang, G.~Sun, Z.~Sun, J.~Wang, J.~Li, C.~Zhao, J.~Wu, S.~Liang, M.~Yin, P.~Wang \emph{et~al.}, ``Toward realization of low-altitude economy networks: Core architecture, integrated technologies, and future directions,'' \emph{arXiv preprint arXiv:2504.21583}, 2025.

\bibitem{mcmahan2017communication}
B.~McMahan, E.~Moore, D.~Ramage, S.~Hampson, and B.~A. y~Arcas, ``Communication-efficient learning of deep networks from decentralized data,'' in \emph{Artificial intelligence and statistics}.\hskip 1em plus 0.5em minus 0.4em\relax PMLR, 2017, pp. 1273--1282.

\bibitem{yu2025lightweight}
B.~Yu, J.~Zhao, K.~Zhang, J.~Gong, and H.~Qian, ``Lightweight and dynamic privacy-preserving federated learning via functional encryption,'' \emph{IEEE Transactions on Information Forensics and Security}, 2025.

\bibitem{voigt2017eu}
P.~Voigt and A.~Von~dem Bussche, ``The eu general data protection regulation (gdpr),'' \emph{A practical guide, 1st ed., Cham: Springer International Publishing}, vol.~10, no. 3152676, pp. 10--5555, 2017.

\bibitem{harding2019understanding}
E.~L. Harding, J.~J. Vanto, R.~Clark, L.~Hannah~Ji, and S.~C. Ainsworth, ``Understanding the scope and impact of the california consumer privacy act of 2018,'' \emph{Journal of Data Protection \& Privacy}, vol.~2, no.~3, pp. 234--253, 2019.

\bibitem{liu2021federaser}
G.~Liu, X.~Ma, Y.~Yang, C.~Wang, and J.~Liu, ``Federaser: Enabling efficient client-level data removal from federated learning models,'' in \emph{2021 IEEE/ACM 29th international symposium on quality of service (IWQOS)}.\hskip 1em plus 0.5em minus 0.4em\relax IEEE, 2021, pp. 1--10.

\bibitem{wu2022federated}
L.~Wu, S.~Guo, J.~Wang, Z.~Hong, J.~Zhang, and Y.~Ding, ``Federated unlearning: Guarantee the right of clients to forget,'' \emph{IEEE Network}, vol.~36, no.~5, pp. 129--135, 2022.

\bibitem{ding2024strategic}
N.~Ding, E.~Wei, and R.~Berry, ``Strategic data revocation in federated unlearning,'' in \emph{IEEE INFOCOM 2024-IEEE Conference on Computer Communications}.\hskip 1em plus 0.5em minus 0.4em\relax IEEE, 2024, pp. 1151--1160.

\bibitem{bourtoule2021machine}
L.~Bourtoule, V.~Chandrasekaran, C.~A. Choquette-Choo, H.~Jia, A.~Travers, B.~Zhang, D.~Lie, and N.~Papernot, ``Machine unlearning,'' in \emph{2021 IEEE Symposium on Security and Privacy (SP)}.\hskip 1em plus 0.5em minus 0.4em\relax IEEE, 2021, pp. 141--159.

\bibitem{shaik2024exploring}
T.~Shaik, X.~Tao, H.~Xie, L.~Li, X.~Zhu, and Q.~Li, ``Exploring the landscape of machine unlearning: A comprehensive survey and taxonomy,'' \emph{IEEE Transactions on Neural Networks and Learning Systems}, 2024.

\bibitem{ginart2019making}
A.~Ginart, M.~Guan, G.~Valiant, and J.~Y. Zou, ``Making ai forget you: Data deletion in machine learning,'' \emph{Advances in neural information processing systems}, vol.~32, 2019.

\bibitem{liu2021revfrf}
Y.~Liu, Z.~Ma, Y.~Yang, X.~Liu, J.~Ma, and K.~Ren, ``Revfrf: Enabling cross-domain random forest training with revocable federated learning,'' \emph{IEEE Transactions on Dependable and Secure Computing}, vol.~19, no.~6, pp. 3671--3685, 2021.

\bibitem{guo2023fast}
X.~Guo, P.~Wang, S.~Qiu, W.~Song, Q.~Zhang, X.~Wei, and D.~Zhou, ``Fast: Adopting federated unlearning to eliminating malicious terminals at server side,'' \emph{IEEE Transactions on Network Science and Engineering}, 2023.

\bibitem{fu2024client}
C.~Fu, W.~Jia, and N.~Ruan, ``Client-free federated unlearning via training reconstruction with anchor subspace calibration,'' in \emph{ICASSP 2024-2024 IEEE International Conference on Acoustics, Speech and Signal Processing (ICASSP)}.\hskip 1em plus 0.5em minus 0.4em\relax IEEE, 2024, pp. 9281--9285.

\bibitem{cao2023fedrecover}
X.~Cao, J.~Jia, Z.~Zhang, and N.~Z. Gong, ``Fedrecover: Recovering from poisoning attacks in federated learning using historical information,'' in \emph{2023 IEEE Symposium on Security and Privacy (SP)}.\hskip 1em plus 0.5em minus 0.4em\relax IEEE, 2023, pp. 1366--1383.

\bibitem{yuan2023federated}
W.~Yuan, H.~Yin, F.~Wu, S.~Zhang, T.~He, and H.~Wang, ``Federated unlearning for on-device recommendation,'' in \emph{Proceedings of the sixteenth ACM international conference on web search and data mining}, 2023, pp. 393--401.

\bibitem{zhang2023fedrecovery}
L.~Zhang, T.~Zhu, H.~Zhang, P.~Xiong, and W.~Zhou, ``Fedrecovery: Differentially private machine unlearning for federated learning frameworks,'' \emph{IEEE Transactions on Information Forensics and Security}, 2023.

\bibitem{romandini2024federated}
N.~Romandini, A.~Mora, C.~Mazzocca, R.~Montanari, and P.~Bellavista, ``Federated unlearning: A survey on methods, design guidelines, and evaluation metrics,'' \emph{IEEE Transactions on Neural Networks and Learning Systems}, 2024.

\bibitem{sommer2022athena}
D.~M. Sommer, L.~Song, S.~Wagh, and P.~Mittal, ``Athena: Probabilistic verification of machine unlearning,'' \emph{Proceedings on Privacy Enhancing Technologies}, 2022.

\bibitem{guo2023verifying}
Y.~Guo, Y.~Zhao, S.~Hou, C.~Wang, and X.~Jia, ``Verifying in the dark: Verifiable machine unlearning by using invisible backdoor triggers,'' \emph{IEEE Transactions on Information Forensics and Security}, 2023.

\bibitem{gao2024verifi}
X.~Gao, X.~Ma, J.~Wang, Y.~Sun, B.~Li, S.~Ji, P.~Cheng, and J.~Chen, ``Verifi: Towards verifiable federated unlearning,'' \emph{IEEE Transactions on Dependable and Secure Computing}, 2024.

\bibitem{wang2024fedu}
W.~Wang, C.~Zhang, Z.~Tian, and S.~Yu, ``Fedu: Federated unlearning via user-side influence approximation forgetting,'' \emph{IEEE Transactions on Dependable and Secure Computing}, 2024.

\bibitem{pan2025federated}
Z.~Pan, Z.~Wang, C.~Li, K.~Zheng, B.~Wang, X.~Tang, and J.~Zhao, ``Federated unlearning with gradient descent and conflict mitigation,'' in \emph{Proceedings of the AAAI Conference on Artificial Intelligence}, vol.~39, no.~19, 2025, pp. 19\,804--19\,812.

\bibitem{liu2024survey}
Z.~Liu, Y.~Jiang, J.~Shen, M.~Peng, K.-Y. Lam, X.~Yuan, and X.~Liu, ``A survey on federated unlearning: Challenges, methods, and future directions,'' \emph{ACM Computing Surveys}, vol.~57, no.~1, pp. 1--38, 2024.

\bibitem{rivest1978data}
R.~L. Rivest, L.~Adleman, M.~L. Dertouzos \emph{et~al.}, ``On data banks and privacy homomorphisms,'' \emph{Foundations of secure computation}, vol.~4, no.~11, pp. 169--180, 1978.

\bibitem{paillier1999public}
P.~Paillier, ``Public-key cryptosystems based on composite degree residuosity classes,'' in \emph{International conference on the theory and applications of cryptographic techniques}.\hskip 1em plus 0.5em minus 0.4em\relax Springer, 1999, pp. 223--238.

\bibitem{bellare1994incremental}
M.~Bellare, O.~Goldreich, and S.~Goldwasser, ``Incremental cryptography: The case of hashing and signing,'' in \emph{Advances in Cryptology—CRYPTO’94: 14th Annual International Cryptology Conference Santa Barbara, California, USA August 21--25, 1994 Proceedings 14}.\hskip 1em plus 0.5em minus 0.4em\relax Springer, 1994, pp. 216--233.

\bibitem{deng2012mnist}
L.~Deng, ``The mnist database of handwritten digit images for machine learning research [best of the web],'' \emph{IEEE signal processing magazine}, vol.~29, no.~6, pp. 141--142, 2012.

\bibitem{xiao2017fashion}
H.~Xiao, K.~Rasul, and R.~Vollgraf, ``Fashion-mnist: a novel image dataset for benchmarking machine learning algorithms,'' \emph{arXiv preprint arXiv:1708.07747}, 2017.

\bibitem{cohen2017emnist}
G.~Cohen, S.~Afshar, J.~Tapson, and A.~Van~Schaik, ``Emnist: Extending mnist to handwritten letters,'' in \emph{2017 international joint conference on neural networks (IJCNN)}.\hskip 1em plus 0.5em minus 0.4em\relax IEEE, 2017, pp. 2921--2926.

\bibitem{wortsman2020supermasks}
M.~Wortsman, V.~Ramanujan, R.~Liu, A.~Kembhavi, M.~Rastegari, J.~Yosinski, and A.~Farhadi, ``Supermasks in superposition,'' \emph{Advances in Neural Information Processing Systems}, vol.~33, pp. 15\,173--15\,184, 2020.

\bibitem{jia2021proof}
H.~Jia, M.~Yaghini, C.~A. Choquette-Choo, N.~Dullerud, A.~Thudi, V.~Chandrasekaran, and N.~Papernot, ``Proof-of-learning: Definitions and practice,'' in \emph{2021 IEEE Symposium on Security and Privacy (SP)}.\hskip 1em plus 0.5em minus 0.4em\relax IEEE, 2021, pp. 1039--1056.

\bibitem{garg2023experimenting}
S.~Garg, A.~Goel, S.~Jha, S.~Mahloujifar, M.~Mahmoody, G.-V. Policharla, and M.~Wang, ``Experimenting with zero-knowledge proofs of training,'' in \emph{Proceedings of the 2023 ACM SIGSAC Conference on Computer and Communications Security}, 2023, pp. 1880--1894.

\bibitem{cao2015towards}
Y.~Cao and J.~Yang, ``Towards making systems forget with machine unlearning,'' in \emph{2015 IEEE symposium on security and privacy}.\hskip 1em plus 0.5em minus 0.4em\relax IEEE, 2015, pp. 463--480.

\bibitem{li2025machine}
N.~Li, C.~Zhou, Y.~Gao, H.~Chen, Z.~Zhang, B.~Kuang, and A.~Fu, ``Machine unlearning: Taxonomy, metrics, applications, challenges, and prospects,'' \emph{IEEE Transactions on Neural Networks and Learning Systems}, 2025.

\bibitem{golatkar2020eternal}
A.~Golatkar, A.~Achille, and S.~Soatto, ``Eternal sunshine of the spotless net: Selective forgetting in deep networks,'' in \emph{Proceedings of the IEEE/CVF Conference on Computer Vision and Pattern Recognition}, 2020, pp. 9304--9312.

\bibitem{mehta2022deep}
R.~Mehta, S.~Pal, V.~Singh, and S.~N. Ravi, ``Deep unlearning via randomized conditionally independent hessians,'' in \emph{Proceedings of the IEEE/CVF Conference on Computer Vision and Pattern Recognition}, 2022, pp. 10\,422--10\,431.

\bibitem{jiang2025towards}
Y.~Jiang, J.~Shen, Z.~Liu, C.~W. Tan, and K.-Y. Lam, ``Towards efficient and certified recovery from poisoning attacks in federated learning,'' \emph{IEEE Transactions on Information Forensics and Security}, 2025.

\bibitem{halimi2022federated}
A.~Halimi, S.~Kadhe, A.~Rawat, and N.~Baracaldo, ``Federated unlearning: How to efficiently erase a client in fl?'' \emph{arXiv preprint arXiv:2207.05521}, 2022.

\bibitem{li2023subspace}
G.~Li, L.~Shen, Y.~Sun, Y.~Hu, H.~Hu, and D.~Tao, ``Subspace based federated unlearning,'' \emph{arXiv preprint arXiv:2302.12448}, 2023.

\bibitem{wang2024server}
P.~Wang, W.~Song, H.~Qi, C.~Zhou, F.~Li, Y.~Wang, P.~Sun, and Q.~Zhang, ``Server-initiated federated unlearning to eliminate impacts of low-quality data,'' \emph{IEEE Transactions on Services Computing}, 2024.

\bibitem{xu2019verifynet}
G.~Xu, H.~Li, S.~Liu, K.~Yang, and X.~Lin, ``Verifynet: Secure and verifiable federated learning,'' \emph{IEEE Transactions on Information Forensics and Security}, vol.~15, pp. 911--926, 2019.

\bibitem{zheng2022aggregation}
Y.~Zheng, S.~Lai, Y.~Liu, X.~Yuan, X.~Yi, and C.~Wang, ``Aggregation service for federated learning: An efficient, secure, and more resilient realization,'' \emph{IEEE Transactions on Dependable and Secure Computing}, vol.~20, no.~2, pp. 988--1001, 2022.

\bibitem{wang2022vosa}
Y.~Wang, A.~Zhang, S.~Wu, and S.~Yu, ``Vosa: Verifiable and oblivious secure aggregation for privacy-preserving federated learning,'' \emph{IEEE Transactions on Dependable and Secure Computing}, vol.~20, no.~5, pp. 3601--3616, 2022.

\bibitem{fu2020vfl}
A.~Fu, X.~Zhang, N.~Xiong, Y.~Gao, H.~Wang, and J.~Zhang, ``Vfl: A verifiable federated learning with privacy-preserving for big data in industrial iot,'' \emph{IEEE Transactions on Industrial Informatics}, vol.~18, no.~5, pp. 3316--3326, 2020.

\bibitem{hahn2021versa}
C.~Hahn, H.~Kim, M.~Kim, and J.~Hur, ``Versa: Verifiable secure aggregation for cross-device federated learning,'' \emph{IEEE Transactions on Dependable and Secure Computing}, vol.~20, no.~1, pp. 36--52, 2021.

\bibitem{weng2024proof}
J.~Weng, S.~Yao, Y.~Du, J.~Huang, J.~Weng, and C.~Wang, ``Proof of unlearning: Definitions and instantiation,'' \emph{IEEE Transactions on Information Forensics and Security}, 2024.

\bibitem{sommer2020towards}
D.~M. Sommer, L.~Song, S.~Wagh, and P.~Mittal, ``Towards probabilistic verification of machine unlearning,'' \emph{arXiv preprint arXiv:2003.04247}, 2020.

\end{thebibliography}

\end{document}